\documentclass[11pt,fleqn]{article}

\usepackage{amsmath,amssymb,amsthm,enumerate,cite}

\allowdisplaybreaks

\newcommand{\p}{\partial}

\newcommand{\sign}{\mathop{\rm sign}\nolimits}

\newcommand{\todo}[1][\null]{\ensuremath{\clubsuit}}

\newcommand{\noprint}[1]{}
\newcommand{\checked}[1][\null]{\ensuremath{\boldsymbol{\surd}}}

\newcounter{tbn}

\newcounter{mcasenum}
\renewcommand{\themcasenum}{{\rm\arabic{mcasenum}}}

\newtheorem{theorem}{Theorem}
\newtheorem{lemma}{Lemma}
\newtheorem{corollary}{Corollary}
\newtheorem{proposition}{Proposition}
\newtheorem*{proposition*}{Proposition}
{\theoremstyle{definition}

\newtheorem{remark}{Remark}
}

\begin{document}

\par\noindent {\LARGE\bf
Extended Group Analysis\\ of Variable Coefficient
Reaction--Diffusion Equations\\  with Exponential Nonlinearities
\par}

{\vspace{4mm}\par\noindent\large O.\,O.~Vaneeva$^{\dag 1}$, R.\,O.~Popovych$^{\dag\ddag 2}$ and C. Sophocleous$^{\S 3}$
\par\vspace{2mm}\par}

{\vspace{2mm}\par\noindent\it
${}^\dag$\ Institute of Mathematics of National Academy of Sciences of Ukraine, \\
$\phantom{{}^\dag}$\ 3 Tereshchenkivska Str., Kyiv-4, 01601 Ukraine\\[2mm]
$^\ddag$\;Wolfgang Pauli Institute, Nordbergstra{\ss}e 15, A-1090 Wien, Austria\\[2mm]
$^\S$\;Department of Mathematics and Statistics, University of Cyprus, Nicosia CY 1678, Cyprus
}

{\par\vspace{2mm}\par\noindent
$\phantom{{}^\dag{}\;}$E-mail: \it$^1$vaneeva@imath.kiev.ua,
$^2$rop@imath.kiev.ua,
$^3$christod@ucy.ac.cy
\par}

{\vspace{5mm}\par\noindent\hspace*{8mm}\parbox{140mm}{\small
The group classification of a class of variable coefficient reaction--diffusion equations with exponential nonlinearities
is carried out up to both the equivalence generated by the corresponding generalized equivalence group and the general point equivalence.
The set of admissible transformations of this class is exhaustively described
via finding the complete family of maximal normalized subclasses and associated conditional equivalence groups.
Limit processes between variable coefficient reaction--diffusion equations with power nonlinearities and those with exponential nonlinearities
are simultaneously studied with limit processes between objects related to these equations (including Lie symmetries, exact solutions and conservation laws).
}\par\vspace{5mm}}

\section{Introduction}

The problem of extended group symmetry analysis of variable coefficient reaction--diffusion equations of the general form
\begin{equation}\label{equation_RD_general}
f(x)u_t=\left(g(x)A(u)u_x\right)_x+h(x)B(u),
\end{equation}
where $fgA\neq0$, was initiated in~\cite{VJPS2007,Vaneeva&Popovych&Sophocleous2009}.
The case of $A$ and $B$ being power functions, i.e., the class of equations having the form
\begin{equation}\label{eqRDfghTwoPower}
f(x)u_t=(g(x)u^nu_x)_x+h(x)u^m,
\end{equation}
where $fg\neq0$ and $(n,mh)\ne(0,0)$, was successfully investigated.
For many reasons, the natural continuation of that study is to consider equations of the form~\eqref{equation_RD_general} with $A$ and $B$ being exponential functions,
\begin{equation} \label{eqRDfghExp}
f(x)u_t=(g(x)e^{nu}u_x)_x+h(x)e^{mu}.
\end{equation}
Here $f=f(x)$, $g=g(x)$ and $h=h(x)$  are arbitrary smooth functions of the variable~$x$,
$fg\neq0$ and $n$ and $m$ are arbitrary constants.
The linear case, which is singled out by the condition $n=m=0$, is excluded from consideration as it is well investigated.
The semilinear equations of the form~\eqref{eqRDfghExp}, which correspond to the constraints $n=0$ and $m\ne0$, were already considered
in~\cite{vaneeva_proc,Vaneeva&Popovych&Sophocleous2011}.
Moreover, equations of the form~\eqref{eqRDfghExp} with $n\ne0$ are not related to linear and semilinear equations
of the same form via point transformations.
This is why in the present paper we study only the class of equations of the form~\eqref{eqRDfghExp} with $fgn\ne0$,
which we briefly call \emph{class}~\eqref{eqRDfghExp}.
Note that the parameter~$n$ can be gauged to~$1$ by a simple scaling of variables from the very beginning
but we will not use this gauge in the present paper and will deal with the general form~\eqref{eqRDfghExp}.

Motivated by the work in~\cite{VJPS2007} as a general outline for similar studies,
we carry out the complete group classification of class~\eqref{eqRDfghExp}.
In order to achieve this in the easiest way, we apply a number of modern tools of group analysis of differential equations:
generalized extended equivalence groups,
conditional equivalence groups,
maximal normalized subclasses,
the method of furcate split,
variable gauges of arbitrary elements by equivalence transformations,
additional equivalence transformations,
etc.
\cite{Ivanova&Popovych&Sophocleous2006I,Nikitin&Popovych2001,
Popovych&Ivanova2004NVCDCEs,Popovych&Kunzinger&Eshraghi2010,
VJPS2007,Vaneeva&Popovych&Sophocleous2009}.

The structure of this paper is as follows.
At first, in Section~\ref{SectionOnPreliminaryStudyOfAdmTrans} using the direct method
we derive the determining equations for admissible point transformations in class~\eqref{eqRDfghExp}.
In Section~\ref{SectionOnEquivGroups}
we find equivalence groups of different kinds for class~\eqref{eqRDfghExp}
and show that the consideration of this class can be simplified by setting the gauge $g=1$.
Moreover, two possibilities for further simplification in the case $m=n$ via using the associated conditional
equivalence group of class~\eqref{eqRDfghExp} are discussed.
In Section~\ref{SectionOnLieSymmetries}
the group classification of class~\eqref{eqRDfghExp} is carried out up to equivalence generated
by the generalized extended equivalence group of this class.
The usage of the method of furcate split for solving the group classification problem is explained in detail.
Admissible point transformations of equations from class~\eqref{eqRDfghExp} are exhaustively described
in Section~\ref{SectionClassificationOfAdmTrans}
in terms of conditional equivalence groups and normalized subclasses.
Section~\ref{SectionOnContractions} is devoted to the study of
contractions (nontrivial limit processes) between equations
from classes~\eqref{eqRDfghTwoPower} and~\eqref{eqRDfghExp}.
Using the derived contractions and the results obtained for class~\eqref{eqRDfghTwoPower} in~\cite{VJPS2007},
we construct exact solutions and conservation laws for equations from class~\eqref{eqRDfghExp}.
The results of the paper are summed up in the conclusion. 
The thorough gauging of constant parameters in the group classification list 
for class~\eqref{eqRDfghTwoPower}, which was not presented in \cite{VJPS2007}, 
is discussed in the appendix.

\section{Preliminary study of admissible transformations}\label{SectionOnPreliminaryStudyOfAdmTrans}

Due to a special structure of equations from
class~\eqref{eqRDfghExp}, the following problem can be solved
completely. To describe all point transformations each of which
connects a pair of equations from class~\eqref{eqRDfghExp}. Such
transformations are called
\emph{form-preserving}~\cite{Kingston&Sophocleous1998} or
\emph{admissible}~\cite{Popovych&Eshraghi2005} or \emph{allowed}~\cite{Winternitz92}
\emph{transformations}. See~\cite{Popovych&Eshraghi2005} for
stronger definitions. They can be naturally interpreted in terms
of the category theory~\cite{Prokhorova2005}.
Note that there exists an infinitesimal equivalent of this notion~\cite{Borovskikh2004}.

At first we make a preliminary study of admissible transformations for class~\eqref{eqRDfghExp}
using the direct method~\cite{Kingston&Sophocleous1998}.
We briefly describe the corresponding procedure.
In what follows we refer to~\eqref{eqRDfghExp} 
as a single equation under assuming that the arbitrary elements $f$, $g$, $h$, $n$ and~$m$ are fixed  
and as a class of equations if the arbitrary elements are varied. 
We use the same agreement for other classes of equations. 
Consider a pair of equations from the class under consideration,
i.e., equation~\eqref{eqRDfghExp} and the equation
\begin{equation}\label{eqRDfghExptilde}
\tilde f(\tilde x){\tilde u}_{\tilde t}=(\tilde g(\tilde x)e^{\tilde n\tilde u}{\tilde u}_{\tilde x})_{\tilde x}
+\tilde h(\tilde x)e^{\tilde m\tilde u},
\end{equation}
and assume that these equations are connected via a point
transformation~$\mathcal T$  of the general form
\[
\tilde t=T(t,x,u), \quad \tilde x=X(t,x,u), \quad \tilde u=U(t,x,u),
\]
where $|\partial(T,X,U)/\partial(t,x,u)|\ne0$.
We have to derive the determining equations for the functions~$T$, $X$ and $U$ and then to solve them.
Simultaneously we have to find the connection between arbitrary elements
of equations~\eqref{eqRDfghExp} and~\eqref{eqRDfghExptilde}. 

After substituting the expressions for the new variables (with tildes)
into~\eqref{eqRDfghExptilde}, we obtain an equation in the old variables (without tildes).
It should be an identity on the manifold~$\mathcal L$ determined by~\eqref{eqRDfghExp}
in the second-order jet space~$J^2$ with the independent variables $(t,x)$ and the dependent variable~$u$.
To involve the constraint between variables of~$J^2$ on the manifold~$\mathcal L$,
we substitute the expression of~$u_t$ implied by equation~\eqref{eqRDfghExp}.
The splitting of this identity with respect to the derivatives $u_{tx}$, $u_{tt}$, $u_{xx}$ and $u_x$
implies the determining equations for the functions~$T$, $X$ and $U$.
In particular, we obtain that $T_u=T_x=X_u=0$ (and hence $T_tX_xU_u\ne0$).
This well agrees with results on more general classes of evolution equations
\cite{Kingston&Sophocleous1998,Magadeev1993,Popovych&Ivanova2004NVCDCEs,Prokhorova2005}.
The restrictions $T_u=T_x=0$ are common for point transformations between $(1+1)$-dimensional evolution equations of order greater than one.
The restriction $X_u=0$ follows from the additional condition that the related equations are quasilinear, i.e., they have the form
$u_t=F(t,x,u)u_{xx}+G(t,x,u,u_x)$ with $F\ne 0$.

We take into account the above determining equations.
Then collecting the coefficients of~$u_{xx}$ gives
$\tilde f gX_x^2e^{nu}=f\tilde g T_te^{\tilde n U}.$
Dividing the last equation by $e^{nu}$ and differentiating the
resulting equation with respect to~$u$,
we obtain that $f\tilde g T_te^{\tilde n U-nu}(\tilde n U_u-n)=0$, i.e., $\tilde nU_u=n$.
Therefore, $U_u$ is a nonvanishing constant, which we denote by~$\delta_3$,
and the component~$U$ can be represented in the form $U=\delta_3 u+\psi(t,x)$
with the smooth function~$\psi$ of~$t$ and~$x$.
The collecting coefficients of~$u_x$ leads to the equation
\[
\left(2\tilde g V_x X_x-\frac{\tilde f}f\,\frac{X_ x^{\,3}}{T_t}g_x+{\tilde g}_{\tilde x}V X_x^{\,2}
-{\tilde g}\,V X_{xx}\right)T_t e^{nu}+{\tilde f}X_x^{\,2}X_t =0,
\]
where by $V$ we denote $e^{\frac{n}{\delta_3} \psi}$, which is a nonvanishing function of~$t$ and~$x$.
The further splitting with respect to~$u$ implies, in particular, that $X_t=0$.
Taking into account the found constraints, we deduce the following expressions for components of the point transformation~$\mathcal T$:
\begin{equation}\label{adm_tr}
T=T(t),\quad X=\varphi(x),\quad U=\delta_3 u+\psi(t,x),\quad \tilde n=\frac{n}{\delta_3},
\end{equation}
where $T$, $\varphi$ and~$\psi$ are smooth functions of their arguments, $\delta_3$ is a constant and $T_t\varphi_x\delta_3\neq0$.
The remainder of determining equations for~$\mathcal T$ is given by
the unused constraints associated with the coefficients~$u_{xx}$ and $u_x$ and
the collection of the term without derivatives of~$u$,
\begin{gather}\label{det_eq_adm_tr1}
f\tilde gVT_t-\tilde f g \varphi_x^{\,2}=0,\quad
2\frac{V_x}V-\frac{g_x}g+\frac{\tilde g_{\tilde x}}{\tilde g} \varphi_x-\frac{\varphi_{xx}}{\varphi_x}=0,
\\\label{det_eq_adm_tr2}
 \frac{T_t}{\varphi_x}\left(\tilde g\frac{ V_x}{\varphi_x}\right)_xe^{nu}
-\tilde f\frac{V_t}{V}-n\frac{\tilde f}fhe^{mu}+
\frac n{\delta_3}T_te^{\tilde m  \psi}\tilde h e^{\tilde m \delta_3 u}=0.
\end{gather}
(We additionally simplify the above system by substituting
the expression for the ratio $\tilde f/f$ from the first equation into the second one.)

\section{Equivalence groups}\label{SectionOnEquivGroups}

In order to solve a group classification problem
it is essential to derive the equivalence transformations which
preserve the differential structure of the class of differential equations under consideration and transform only
its arbitrary elements.
These transformations form a group~\cite{Ovsiannikov1982}, which is called the equivalence group of the class.

There exist several kinds of equivalence groups. The usual equivalence group
consists of the nondegenerate point transformations of independent and dependent variables
as well as arbitrary elements of the class.
Moreover, transformations of independent and dependent variables do not depend on arbitrary
elements. If such dependence arises then the corresponding equivalence group is called {\it generalized} \cite{Meleshko1994}.
If new arbitrary elements are expressed via old ones in some non-point, possibly nonlocal, way (e.g. new
arbitrary elements are determined via integrals of old ones) then
the equivalence transformations are called {\it extended}.
The first examples of a generalized equivalence group and of an extended equivalence group
were presented in~\cite{Meleshko1994} and~\cite{mogran}, respectively.
See a number of examples
on different equivalence groups and their role in solving complicated group classification problems,
e.g., in~\cite{Vaneeva&Popovych&Sophocleous2009,VJPS2007,Ivanova&Popovych&Sophocleous2006I}.

It appears that the generalized extended equivalence group~$\hat G^{\sim}$ of class~\eqref{eqRDfghExp}
is nontrivial, i.e., it is wider than the usual equivalence group of the same class.
We construct the group~$\hat G^{\sim}$ using the constraints derived
in the course of the preliminary study of admissible transformations
and varying arbitrary elements in equations~\eqref{det_eq_adm_tr1} and~\eqref{det_eq_adm_tr2},
which form the complete system of classifying conditions for admissible transformations in class~\eqref{eqRDfghExp}.

For the general values of~$n$ and~$m$ the exponents~$e^{0u}$, $e^{nu}$ and~$e^{mu}$ are linearly independent.
This is why due to varying of arbitrary elements we necessarily have that $\tilde m=m/\delta_3$.
Splitting of equation~\eqref{det_eq_adm_tr2} with respect to $u$ results in $V_t=0$ (i.e., $\psi(t,x)=\psi(x)$) and therefore differentiating
both sides of the first equation of~\eqref{det_eq_adm_tr1}  with respect to $t$ we have
$T_{tt}=0,$ i.e., $T=\delta_1t+\delta_2$.
Then solving equations~\eqref{det_eq_adm_tr1} we obtain
\[
\tilde g=\frac{\delta_0\varphi_x}{V^2}\,g,\quad \tilde f=\frac{\delta_0T_t}{V\varphi_x}\,f,
\]
where $\delta_0$ is a nonvanishing constant.
The remaining equations are
\begin{gather*}
\left(\tilde g\frac{V_x}{\varphi_x}\right)_x=0,\quad
\tilde h fe^{\frac m{\delta_3}\psi}\delta_1 - {\delta_3}\tilde fh=0.
\end{gather*}
Solving these equations, we obtain
\begin{gather}
\tilde h=\frac{\delta_0\delta_3}{V\varphi_x}e^{-\frac m{\delta_3}  \psi}h,\quad V=\left(\delta_4\int \frac{dx}{g(x)}+\delta_5\right)^{-1},
\end{gather}
where $\delta_4$ and $\delta_5$ are arbitrary constants with $(\delta_4,\delta_5)\ne(0,0)$.

The equivalence group of the subclass singled out from class~\eqref{eqRDfghExp} by the constraint $m=n$
is wider than the equivalence group of the whole class~\eqref{eqRDfghExp}.
In other words, the constraint $m=n$ is associated with a nontrivial conditional equivalence group of class~\eqref{eqRDfghExp}.
In this case we have less splitting with respect to~$u$. As a result,
$V$ (or, equivalently, $\psi$) is an arbitrary smooth function of~$t$ and
\[
\tilde h=\frac{\delta_0\delta_3}{nV\varphi_x}\left[\frac{nh}V-\left(\frac{gV_x}{V^2}\right)_x\,\right].
\]

The above results are summarized in Theorems~\ref{equivfgh_exp} and~\ref{equivfgh_exp_m=n}.
In what follows we use the notation $\Psi=1/V=e^{-\frac n{\delta_3}\psi}$ in order to simplify formulas. 

\begin{theorem}\label{equivfgh_exp}
The generalized extended equivalence group~$\hat G^{\sim}$ of
class~\eqref{eqRDfghExp} is formed by the transformations
\[
\begin{array}{l}
\tilde t=\delta_1 t+\delta_2,\quad \tilde x=\varphi(x), \quad
\tilde u=\delta_3u+\psi(x), \\[1ex]
\tilde f=\dfrac{\delta_0\delta_1}{\varphi_x}\Psi f, \quad
\tilde g=\delta_0\varphi_x\Psi^2 g, \quad
\tilde h=\dfrac{\delta_0\delta_3}{\varphi_x} e^{-\frac m{\delta_3}\psi}\Psi h, \quad
\tilde n=\dfrac{n}{\delta_3}, \quad
\tilde m=\dfrac{m}{\delta_3},
\end{array}
\]
where $\varphi$ is an arbitrary smooth function
of~$x$, $\varphi_x\not=0$; $\psi$ and $\Psi$ are determined by the formulas
\[
\psi(x)= -\dfrac{\delta_3}{n}\ln\left|\Psi(x)\right|,\quad\Psi(x)=\delta_4\int \frac{dx}{g(x)}+\delta_5.
\]
$\delta_j,$ $j=0,\dots,5$, are arbitrary constants,
$\delta_0\delta_1\delta_3\not=0$ and $(\delta_4,\delta_5)\ne(0,0)$.
\end{theorem}

Thus, elements of~$\hat G^{\sim}$ are parameterized by six arbitrary constants and a single arbitrary smooth function of~$x$.
The usual equivalence group $G^{\sim}$ of class~\eqref{eqRDfghExp} is the subgroup of the generalized extended
equivalence group $\hat G^{\sim}$, which is singled out with the condition $\delta_4=0$.

\begin{theorem}\label{equivfgh_exp_m=n}
The generalized extended equivalence group of the class of equations
\begin{gather}\label{eqRDfghExp_m=n}
f(x)u_t=(g(x)e^{nu}u_x)_x+h(x)e^{nu}\quad\mbox{with}\quad nfg\ne0,
\end{gather}
coincides with the usual equivalence group~$G^{\sim}_{m=n}$ of this class and consists of the transformations
\[
\begin{array}{l}
\tilde t=\delta_1 t+\delta_2,\quad \tilde x=\varphi(x),\quad
\tilde u=\delta_3u+\psi(x), \\[2ex]
\tilde f=\dfrac{\delta_0\delta_1}{\varphi_x}\Psi f, \quad
\tilde g=\delta_0\varphi_x \Psi^2 g,\quad
\tilde h=\dfrac{\delta_0\delta_3}{n\varphi_x}\left[nh\Psi+
(g\Psi_x)_x\right]\Psi,
\quad \tilde n=\dfrac{n}{\delta_3},
\end{array}
\]
where $\delta_j,$ $j=0,1,2,3,$ are arbitrary constants,
$\delta_0\delta_1\delta_3\not=0$, $\varphi$ and $\psi$ are arbitrary smooth functions of~$x$ with $\varphi_x\ne0$,
$\Psi(x)=e^{-\frac{n\psi(x)}{\delta_3}}$.
\end{theorem}

Elements of~$G^{\sim}_{m=n}$ are parameterized by four arbitrary constants
and two arbitrary smooth functions of the single variable~$x$.
Therefore, the group~$G^{\sim}_{m=n}$ is really a nontrivial conditional equivalence group of class~\eqref{eqRDfghExp}.

In view of Theorem~\ref{equivfgh_exp}, the family of equivalence transformations
\begin{equation}\label{EqGaugeTransformationForG1}
\tilde t=t,\quad \tilde x=\int^x_{x_0}\frac{dy}{g(y)}, \quad \tilde u=u,
\end{equation}
parameterized by the arbitrary element~$g$ maps class~\eqref{eqRDfghExp} onto its subclass
consisting of equations of the form
\begin{equation*}
\tilde f \tilde u_{\tilde t}=(e^{\tilde n\tilde u}\tilde u_{\tilde x})_{\tilde x}+\tilde h e^{\tilde m\tilde u},
\end{equation*}
with $\tilde g=1$.
The new arbitrary elements are expressed via the old ones in the following way:
\[
\tilde f=fg,\quad
\tilde h=gh, \quad
\tilde m=m,\quad \tilde n =n.
\]
Hence, within the framework of symmetry analysis
it suffices, without loss of generality, to investigate the equations of the general form
\begin{equation} \label{eqRDfghExp2}
f(x)u_t=(e^{nu}u_x)_x+h(x)e^{mu}\quad\mbox{with}\quad nf\ne0
\end{equation}
instead of class~\eqref{eqRDfghExp}
because all results on symmetries, solutions and conservation laws of equations from subclass~\eqref{eqRDfghExp2}
can be extended to the entire class~\eqref{eqRDfghExp} with transformations~\eqref{EqGaugeTransformationForG1}.
In other words, up to $\hat G^{\sim}$-equivalence we can assign the gauge $g=1$ for the arbitrary element~$g$.%
\footnote{Different simple gauges of arbitrary elements by equivalence transformations are possible.
For example, instead of $g=1$ we can set $f=1$ but the gauge $g=1$ is more convenient
because it results in simpler group classification of class~\eqref{eqRDfghExp}.
Simultaneously, we can assign the value~$1$ to the constant arbitrary element~$n$.
It has been noted in the introduction that this gauge is not essential in the course of group classification
and hence it will be used only in the presentation of the final classification list.}

The description of the generalized equivalence group of class~\eqref{eqRDfghExp2}
and the generalized conditional equivalence group of the same class
which is associated with the condition $m=n$ is deduced
from Theorems~1 and~2 by setting $\tilde g=g=1$.

\begin{theorem}\label{equivfgh_exp1}
The generalized equivalence group~$\hat G^{\sim}_1$ of
class~\eqref{eqRDfghExp2} is formed by the transformations
\[
\begin{array}{l}
\tilde t=\delta_1 t+\delta_2,\quad \tilde x=\dfrac{\delta_6
x+\delta_7}{\delta_4 x+\delta_5}, \quad
\tilde u=\delta_3 u-\dfrac{\delta_3}{n}\ln|\delta_4x+\delta_5|, \\[2ex]
\tilde f=\dfrac{\delta_1}{\Delta^2}|\delta_4 x+\delta_5|^3 f, \quad
\tilde h=\dfrac{\delta_3}{\Delta^2}|\delta_4 x+\delta_5|^{\frac mn+3} h, \quad
\tilde n=\dfrac{n}{\delta_3}, \quad
\tilde m=\dfrac{m}{\delta_3},
\end{array}
\]
where  $\delta_j,$ $j=1,\dots,7$, are arbitrary constants such that
$\delta_1\delta_3\not=0$, $\Delta=\delta_5\delta_6-\delta_4\delta_7\ne0$ and
the tuple $(\delta_4,\delta_5,\delta_6,\delta_7)$ is defined up to a nonzero multiplier;
e.g., we can set $\Delta=\pm1$.
\end{theorem}

\begin{theorem}\label{equivfgh_exp_m=n=1}
The class of equations
\begin{gather}\label{eqRDfghExp_m=n=1}
f(x)u_t=(e^{nu}u_x)_x+h(x)e^{nu}\quad\mbox{with}\quad nf\ne0
\end{gather}
admits the generalized equivalence group~$\hat G^{\sim}_{1,m=n}$ consisting of the transformations
\begin{gather}\label{EqEquivTr_m-n}
\tilde t=\delta_1 t+\delta_2,\quad \tilde x=\varphi(x),\quad
\tilde u=\delta_3u+\dfrac{\delta_3}{2n}\ln|\delta_0^2\varphi_x|, \\\nonumber
\tilde f=\delta_0\delta_1|\varphi_x|^{-\frac32} f, \quad
\tilde h=\delta_3\varphi_x^{-2}h+\dfrac{\delta_3}{n}|\varphi_x|^{-\frac32}
(|\varphi_x|^{-\frac12})_{xx},\quad  \tilde n=\dfrac{n}{\delta_3},
\end{gather}
where $\delta_j,$ $j=0,\dots,3,$ are arbitrary constants with $\delta_0\delta_1\delta_3\not=0$
and $\varphi=\varphi(x)$ is an arbitrary smooth function with $\varphi_x\ne0$.
\end{theorem}

Class~\eqref{eqRDfghExp_m=n=1} can be mapped onto a proper subclass with only one arbitrary element depending on~$x$
using an appropriate family of point transformations from the group~$\hat G^{\sim}_{1,m=n}$.
The most convenient gauges for arbitrary elements are the gauges $f=1$ and $h=0$.

The first gauge can be realized by the transformation
\begin{equation}\label{EqGaugeTransformationForf1}
\tilde t=\sign(f)t,\quad \tilde x=\int^x_{x_0}f(y)^\frac23{dy}, \quad
\tilde u=u+\frac1{3n}\ln|f|,
\end{equation}
which maps an equation of the form~\eqref{eqRDfghExp_m=n=1} to the equation
$\tilde u_{\tilde t}=(e^{\tilde n\tilde u}\tilde u_{\tilde x})_{\tilde x}+\tilde h e^{\tilde n\tilde u},$
i.e., $\tilde f=1$.
The other new arbitrary elements are expressed via the old ones in the following way:
\[
\tilde h=
\frac1f\left(f^{-\frac13}h+\frac1{n}(f^{-\frac13})_{xx}\right), \quad
\tilde n =n.\]

\begin{theorem}
The generalized equivalence group~$\hat G^{\sim}_{f=g=1,m=n}$ of the class of equations
\begin{gather}\label{eqRDfghExp_m=n_f=1}
u_t=(e^{nu}u_x)_x+h(x)e^{nu}\quad\mbox{with}\quad n\ne0,
\end{gather}
consists of the transformations
\[
\tilde t=\delta_1 t+\delta_2,\quad \tilde x=\delta_4x+\delta_5,\quad
\tilde u=\delta_3u+\frac{\delta_3}{n}\ln\frac{\delta_4^2}{\delta_1}, \quad
\tilde h=\frac{\delta_3}{\delta_4^2}h, \quad \tilde n=\frac n{\delta_3},
\]
where $\delta_j,$ $j=1,\dots,5$ are arbitrary constants,
$\delta_1\delta_3\delta_4\not=0$, $\delta_1>0$. 
\end{theorem}

The gauge $h=0$ for class~\eqref{eqRDfghExp_m=n=1} can be realized by the family of transformations~\eqref{EqEquivTr_m-n},
where $\delta_1=\delta_3=1,$ $\delta_2=0$ and the function $\varphi$ satisfies the ODE $(|\varphi_x|^{-\frac12})_{xx}+nh|\varphi_x|^{-\frac12}=0.$

\begin{theorem}
The generalized equivalence group~$\hat G^{\sim}_{1,h=0}$ of the class of equations
\begin{gather}\label{eqRDfghExp_m=n_h=0}
f(x)u_t=(e^{nu}u_x)_x\quad\mbox{with}\quad nf\ne0,
\end{gather}
is projection of the group $\hat G^{\sim}_1$ on the space $(t, x, u, f, n).$
\end{theorem}

It is possible to carry out the complete group classification of equations from class~\eqref{eqRDfghExp_m=n=1}
using either the gauge $f=1$ or the gauge $h=0$.
The choice of the first gauge is justified in Section~\ref{SectionOnLieSymmetries}.

\section{Lie symmetries}\label{SectionOnLieSymmetries}

In general, the group classification of class~\eqref{eqRDfghExp} is carried out within the framework of
the classical Lie approach~\cite{Olver1986,Ovsiannikov1982}.
At the same time, we additionally apply a number of modern tools of symmetry analysis,
which essentially simplify both related calculations and the presentation of final results.
Thus, gauging the arbitrary element~$g$ to~$1$ by equivalence transformations, 
we in fact classify subclass~\eqref{eqRDfghExp2} instead of the entire class. 
The main equivalence relation involved in the consideration 
is generated by the generalized equivalence group~$\hat G^{\sim}_1$ of this subclass,
which contains the usual equivalence group of the same subclass as a proper subgroup. 
In other words, we use $\hat G^{\sim}_1$-equivalence, which is stronger than
$G^{\sim}_1$-equivalence prescribed by the classical Lie approach.
Moreover, for group classification of equations from subclass~\eqref{eqRDfghExp2} with $m=n$
we involve the equivalence relation
which is generated by the conditional generalized equivalence group~$\hat G^{\sim}_{1,m=n}$
and is even stronger than $\hat G^{\sim}_1$-equivalence.
In total, this leads to the reduction of classification cases and lowering
the number of additional equivalence transformations to be constructed.
Lie symmetry extensions are separated using
the method of furcate split~\cite{Ivanova&Popovych&Sophocleous2006I,Nikitin&Popovych2001}.

The reduction of group classification of class~\eqref{eqRDfghExp} to that of its subclass~\eqref{eqRDfghExp2} 
is possible due to combining two facts. 
Namely, 
1) each equation from class~\eqref{eqRDfghExp} is mapped by a transformation from 
the  generalized extended equivalence group~$\hat G^{\sim}$ of this class to 
an equation from  subclass~\eqref{eqRDfghExp2} and 
2) the generalized equivalence group~$\hat G^{\sim}_1$ of class~\eqref{eqRDfghExp2} 
only consists of those transformations from~$\hat G^{\sim}$ which preserve the constraint $g=1$. 
The same remark is true for class~\eqref{eqRDfghExp_m=n} with the usual equivalence group~$G^{\sim}_{m=n}$ 
and class~\eqref{eqRDfghExp_m=n=1} with the generalized equivalence group~$\hat G^{\sim}_{1,m=n}$.

We look for vector fields of the form
\begin{equation*}
Q=\tau(t,x,u)\partial_t+\xi(t,x,u)\partial_x+\eta(t,x,u)\partial_u
\end{equation*}
which generate one-parameter groups of point symmetry transformations of an equation~$\mathcal L$ from class~\eqref{eqRDfghExp2}.
These vector fields form the maximal Lie invariance algebra $A^{\max}=A^{\max}(\mathcal L)$ of the equation~$\mathcal L$.
Any such vector field~$Q$ satisfies the criterion of infinitesimal invariance, i.e.,
the action of the second prolongation~$Q^{(2)}$ of~$Q$ on equation~\eqref{eqRDfghExp2}
results in the conditions being an identity for all solutions of this equation. Namely, we require that
\begin{equation}\label{c1}
Q^{(2)}\{f(x)u_t-e^{nu}u_{xx} -ne^{nu}u_x^2 - h(x)e^{mu}\}=0
\end{equation}
identically, modulo equation~\eqref{eqRDfghExp2}.

After the elimination of $u_t$ due to~\eqref{eqRDfghExp2},
equation~(\ref{c1}) becomes an identity in six variables, $t$, $x$, $u$, $u_x$, $u_{xx}$ and $u_{tx}$.
In fact, equation~(\ref{c1}) is a
multivariable polynomial in the variables $u_x$, $u_{xx}$ and $u_{tx}$. The
coefficients of different powers of these variables must be
zero, giving the determining equations on the coefficients $\tau$, $\xi$ and $\eta$.
Solving these equations, we immediately find that
$
\tau =\tau(t),\ \xi =\xi(t,x).
$
This completely agrees with the general results on point transformations between evolution equations~\cite{Kingston&Sophocleous1998}.
Then the remaining determining equations take the form
\begin{gather*}
\xi_tf=(\xi_{xx}-2n\eta_x)e^{nu},\quad
\xi\frac{f_x}{f}=\tau_t-2\xi_x+n\eta,\\
\eta_tf=\eta_{xx}e^{nu}+\left(\xi h_x+\left(2\xi_x+(m-n)\eta\right)h\right)e^{mu}.
\end{gather*}
Splitting of the first equation with respect to~$u$ and the subsequent integration imply that
\[
\xi=\xi(x), \quad\eta=\frac1{2n}\xi_x+\eta^0(t),
\]
where $\eta^0=\eta^0(t)$ is a smooth function of~$t$.

Finally we obtain the \emph{classifying} equations
which essentially include both the residuary uncertainties in the coefficients of the vector field~$Q$
and the arbitrary elements of the class under consideration:
\begin{gather}\label{determ_eq_Lie1}
\xi\frac{f_x}{f}=\tau_t-\frac32\xi_x+n\eta^0,\\\label{determ_eq_Lie2}
\eta^0_tf=\frac{1}{2n}\xi_{xxx}e^{nu}+\left(\xi h_x+\left(\frac{3n+m}{2n}\xi_x+(m-n)\eta^0\right)h\right)e^{mu}.
\end{gather}

In order to find the common part of Lie symmetries for all equations from class~\eqref{eqRDfghExp2} (resp.\ class~\eqref{eqRDfghExp}),
we should vary the arbitrary elements and split with respect to them in the determining equations.
This results in $\tau_t=\xi=\eta=0$.

\begin{proposition}
The kernel algebra, i.e., the intersection of the maximal Lie invariance algebras of equations
from class~\eqref{eqRDfghExp2} (resp.\ class~\eqref{eqRDfghExp}) is the one-dimensional algebra
$A^\cap=\langle\partial_t\rangle$.
\end{proposition}

The classification of possible extensions of~$A^\cap$ is reduced to
the simultaneous solution of equations~\eqref{determ_eq_Lie1} and~\eqref{determ_eq_Lie2}
with respect to both the residuary uncertainties in the coefficients of the vector field~$Q$
and the arbitrary elements up to $\hat G^{\sim}_1$-equivalence.
Further splitting with respect to $u$ in the last equation depends on the value of $m$.
Namely, we should separately consider three exclusive cases
\[
1)\quad m\neq0,n;\quad\quad 2)\quad m=0;\quad\quad 3)\quad m=n.
\]
The case $h=0$ is special as the value of $m$ is undefined in this case
but it can be included to the case $m=n$ (see Section 2).
This is why in the first two cases we assume that $h\neq0$.

\subsection{General case of $\boldsymbol{m}$}\label{SectionOnLieSymmetriesGenCaseOfM}

In the first (general) case splitting of~\eqref{determ_eq_Lie2} gives the equations $\eta^0_t=0$ and $\xi_{xxx}=0$.
Differentiation of~\eqref{determ_eq_Lie1} with respect to $t$ leads to the equation $\tau_{tt}=0.$
Therefore,
\begin{equation}\label{eq_tau_xi_eta}
\tau=c_4t+c_5,\quad \xi=nc_2x^2+c_1x+c_0,\quad \eta=c_2x+c_3.
\end{equation}
 The classifying equations take the form
\begin{gather}\label{eq_lie_sym_f}
(nc_2x^2+c_1x+c_0)\dfrac{f_x}f=-3nc_2x+nc_3+c_4-2c_1,\\\label{eq_lie_sym_h}
(nc_2x^2+c_1x+c_0)\dfrac{h_x}h=-(3n+m)c_2x+(n-m)c_3-2c_1.
\end{gather}

Further investigation will be carried out using the method of furcate
split~\cite{Ivanova&Popovych&Sophocleous2006I,Nikitin&Popovych2001}.
For any operator $Q$ from $A^{\max}$ the substitution of its
coefficients into equations~\eqref{eq_lie_sym_f} and~\eqref{eq_lie_sym_h} gives some equations on $f$ and $h$
of the general form
\begin{gather}\label{eq_lie_sym_fc}
(nax^2+bx+c)\dfrac{f_x}f=-3nax+d,\\\label{eq_lie_sym_hc}
(nax^2+bx+c)\dfrac{h_x}h=-(3n+m)ax+(n-m)p-2b,
\end{gather}
where $a,$ $b,$ $c,$ $d$ and $p$ are constants which are defined up to nonzero multiplier.
The set~$\mathcal V$ of values of the coefficient tuple $(a,b,c,d,p)$
obtained by varying of an operator from $A^{\max}$ is a linear space.
Note that the first three components of any nonzero element from~$\mathcal V$ form a nonzero triple
(otherwise the corresponding system with respect to~$f$ and~$h$ would be inconsistent).
The dimension $k=k(A^{\max})$ of the space~$\mathcal V$ is not greater than 2
otherwise the corresponding equations form an incompatible system with respect to $f$ and $h$.
The value of $k$ is an invariant of the transformations from $\hat G^\sim_1$.
Therefore, there exist three $\hat G^\sim_1$-inequivalent cases
for the value of $k$: $k = 0,$ $k = 1$ and $k = 2$.
We consider these possibilities separately.

\medskip

\noindent{\bf I.}
The condition $k=0$ means that \eqref{eq_lie_sym_f}--\eqref{eq_lie_sym_h} are not equations with respect to $f$ and $h$ but identities.
Therefore, $f$ and $h$ are not constrained and $c_i=0$, $i=0,\dots,4$.
We obtain nothing but the kernel algebra~$A^\cap$ presented by Case~1 of Table~1.

\medskip

\noindent{\bf II.}
If $k=1$ then we have, up to a nonzero multiplier,
exactly one system of the form~\eqref{eq_lie_sym_fc}--\eqref{eq_lie_sym_hc} with respect to the functions $f$ and $h$.
Integrating this system up to $\hat G^\sim_1$-equivalence gives
\[
f=\exp\left({\int\dfrac{-3nax+d}{nax^2+bx+c}\,{\rm d}x}\right),\quad
h=\varepsilon\exp\left({-\int\dfrac{(3n+m)ax+2b+(m-n)p}{nax^2+bx+c}\,{\rm d}x}\right),
\]
where $\varepsilon=\pm1\bmod\hat G^{\sim}_1$, and the coefficients of~$Q$ take the form
\[
\tau=(d+2b-np)\lambda t+c_5,\quad\xi=(nax^2+bx+c)\lambda,\quad\eta=(ax+p)\lambda.
\]
Here $\lambda$ and $c_5$ are arbitrary constants.
Hence $A^{\max}$ is the two-dimensional algebra given by Case~2 of Table~1.

It is necessary to investigate equivalence of cases with the arbitrary elements $f$ and $h$ of the above form,
corresponding to different values of the tuple $(a,b,c,d,p)$.

\begin{lemma}\label{LemmaOntransOfCoeffsOfClassifyingSystem}
Up to $\hat G^{\sim}_1$-equivalence
the parameter tuple~$(a,b,c,d,p)$ can be assumed to belong to the set
$
\{(0,1,0,\bar d,\bar p),\ (0,0,1,1,\check p),\ (0,0,1,0,1),\ (0,0,1,0,0),\ (1/n,0,1,\hat d,\hat p)\},
$
where $\bar d\geqslant-3/2$ and, if $\bar d=-3/2$, $n\bar p\geqslant1/2$;
$\hat d\geqslant0$ and, if $\hat d=0$, $\hat p\geqslant0$.
\end{lemma}

\begin{proof}
Combined with the multiplication by a nonzero constant,
each transformation from the equivalence group~$\hat G^{\sim}_1$ is extended to the coefficient tuple
of system~\eqref{eq_lie_sym_fc}--\eqref{eq_lie_sym_hc}:
\begin{gather*}
\begin{array}{l}
\tilde n\tilde a=\nu(\delta_5^2na-\delta_4\delta_5b+\delta_4^2c),\quad
\tilde b=\nu(-2\delta_5\delta_7na+(\delta_4\delta_7+\delta_5\delta_6)b-2\delta_4\delta_6c),
\\[1ex]
\tilde c=\nu(\delta_7^2na-\delta_6\delta_7b+\delta_6^2c),\quad
\tilde d=\nu\Delta d+3\nu(\delta_5\delta_7na-\delta_4\delta_7b+\delta_4\delta_6c),\\[1ex]
\tilde n\tilde p=\nu\Delta n p-\nu(\delta_5\delta_7na-\delta_4\delta_7b+\delta_4\delta_6c),
\end{array}
\end{gather*}
where $\Delta=\delta_5\delta_6-\delta_4\delta_7\ne0$, $\nu$ is an arbitrary nonzero constant.

There are only three $\hat G^\sim_1$-inequivalent values of the triple $(a,b,c)$
depending on the sign of $D=b^2-4nac$,
\begin{gather*}
(0,1,0)\quad\mbox{if}\quad D>0, \quad
(0,0,1)\quad\mbox{if}\quad  D=0,\quad
(1/n,0,1)\quad\mbox{if}\quad  D<0.
\end{gather*}
Indeed, if $D>0$ then there exist two linearly independent pairs
$(\delta_4,\delta_5)$ and $(\delta_6,\delta_7)$
such that $\delta_5^2na-\delta_4\delta_5b+\delta_4^2c=0$ and $\delta_7^2na-\delta_6\delta_7b+\delta_6^2c=0$.
For these values of $\delta$'s we have $\tilde a=\tilde c=0$.
In the case $D=0$ we choose values of $\delta_4$, $\delta_5$, $\delta_6$ and $\delta_7$
for which $\delta_5^2na-\delta_4\delta_5b+\delta_4^2c=0$ and
the pair $(\delta_6,\delta_7)$ is not proportional to the pair $(\delta_4,\delta_5)$.
Then we obtain that $\tilde a=0$ and
$\tilde b=\nu\delta_7(\delta_4b-2\delta_5na)+\nu\delta_6(\delta_5b-2\delta_4c)=0$.
The residual coefficient ($\tilde b$ if $D>0$ and $\tilde c$ if $D=0$) is necessarily nonzero
and hence can be scaled to~$1$ by choosing the appropriate value of~$\nu$.
If $D<0$, we have $ac\ne0$ and can set $a>0$.
The matrix \[\left(\begin{array}{cc}na&-b/2\\-b/2&c\end{array}\right)\]
is symmetric and positive.
Hence the corresponding bilinear form is a well-defined scalar product.
We choose $\nu=1$ and pairs $(\delta_4,\delta_5)$ and $(\delta_6,\delta_7)$
which are orthonormal with respect to this product.
Then $\tilde n\tilde a=\tilde c=1$ and $\tilde b=0$.

Certain freedom in varying group parameters is preserved even after fixing one of the above inequivalent forms
for both the tuples $(a,b,c)$ and $(\tilde a,\tilde b,\tilde c)$.
This allows us to set constraints for the coefficients~$d$ and~$p$.

Thus, it follows from the equality $(a,b,c)=(\tilde a,\tilde b,\tilde c)=(0,1,0)$
that $\delta_4\delta_5=\delta_6\delta_7=0$ and $\delta_4\delta_7+\delta_5\delta_6=\nu^{-1}$.
There are two cases for the solution of the above system in~$\delta$'s:
either $\delta_4=\delta_7=0$ and $\Delta=\delta_5\delta_6=\nu^{-1}$
or $\delta_5=\delta_6=0$ and $\Delta=-\delta_4\delta_7=-\nu^{-1}$.
In the first case the coefficients~$d$ and~$np$ are identically transformed.
In the second case the transformation takes the form $\tilde d=-d-3$, $\tilde n\tilde p=-n p+1$.
This is why up to $\hat G^{\sim}_1$-equivalence we can assume
that $\tilde d\geqslant-3/2$ and, if $\tilde d=-3/2$, $\tilde n\tilde p\geqslant1/2$.

The equality $(a,b,c)=(\tilde a,\tilde b,\tilde c)=(0,0,1)$ implies $\delta_4=0$, $\nu\delta_6^2=1$
and $\Delta=\delta_5\delta_6\ne0$.
The transformation of $d$ and $np$ is reduced to simultaneous scaling with the same multipliers,
$\tilde d=\nu\Delta d$ and $\tilde n\tilde p=\nu\Delta n p$.
This allows us either to set $\tilde d=1$ if $d\ne0$ or
to scale~$\tilde p$ if $d=0$ and hence $\tilde d=0$. Therefore, we obtain the following inequivalent tuples
 $(0,0,1,1,\check p)$, if $d\neq0,$  $(0,0,1,0,1)$ if $d=0,$ $p\neq0,$ and
 $(0,0,1,0,0)$ in the case $d=p=0.$ The last tuple corresponds to the constant-coefficient equations~\eqref{eqRDfghExp2} (Case 3 of Table 1).

Setting $(na,b,c)=(\tilde n\tilde a,\tilde b,\tilde c)=(1,0,1)$ results in
$\delta_4^2+\delta_5^2=\delta_6^2+\delta_7^2=\nu^{-1}$ and $\delta_4\delta_6+\delta_5\delta_7=0$.
Hence $\delta_6=\varepsilon\delta_5$ and $\delta_7=-\varepsilon\delta_4$, where $\varepsilon=\pm 1$.
The transformation of the coefficients~$d$ and~$np$ is reduced to the multiplication by~$\varepsilon$,
$\tilde d=\varepsilon d$ and $\tilde n\tilde p=\varepsilon n p$.
This is why we can only set
$\hat d\geqslant0$ and, if $\hat d=0$, $\hat p\geqslant0$.
\end{proof}

Lemma~\ref{LemmaOntransOfCoeffsOfClassifyingSystem} implies that
up to $\hat G^\sim_1$-equivalence the case $k=1$ is partitioned into the three inequivalent subcases:
\begin{enumerate}\itemsep=0ex
\item $(f,h)=(|x|^d,\varepsilon |x|^q), \ q=(n-m)p-2\colon\quad$ $\langle\partial_t,\,(d+2-pn)t\partial_t+
x\partial_x+p\partial_u\rangle$;

\item $(f,h)=(e^{dx}, \varepsilon e^{qx}), \ q=(n-m)p\colon\quad$
$\left\langle\partial_t,\,\left(d-pn\right)t\partial_t+\partial_x+p\partial_u\right\rangle$;

\item $(f,h)=\left((x^2+1)^{-\frac32}e^{d\arctan x},\varepsilon (x^2+1)^{-\frac32-\frac m{2n}}e^{q\arctan x}\right), \ q=(n-m)p\colon$\\
$\langle\partial_t,\,(d-pn)t\partial_t+(x^2+1)\partial_x+(x/n+p)\partial_u\rangle$.

\end{enumerate}
Here $(d,q)\neq(0,0),(-3,-3-m/n)$ and $(d,q)\neq(0,0)$ for the first and second subcases, respectively.
(Otherwise $k=2$, see below.)
It additionally follows from Lemma~\ref{LemmaOntransOfCoeffsOfClassifyingSystem} that
up to $\hat G^\sim_1$-equivalence we can set certain constraints for the parameters~$d$ and~$q$.
(It is convenient to use~$q$ instead of~$p$ as a parameter.)
For the first subcase an exhaustive gauge implied by $\hat G^\sim_1$-equivalence
consists of the inequalities $d\geqslant-3/2$ and, if $d=-3/2$, $q\geqslant-3/2-m/(2n)$.
They can be set using the equivalence transformation
\begin{equation}\label{EqEquivTransForExpFHtoConstCoeff}
\tilde t=t,\quad\tilde x=\frac1x,\quad\tilde u=u-\frac1n\ln|x|,
\end{equation}
whose extension to the parameters~$d$ and $q$ is given by $\tilde d=-d-3$ and $\tilde q=-q-3-m/n$.
In the second subcase the parameters~$d$ and $q$ can be gauged using a scaling of $x$.
More precisely, $d=1\bmod\hat G^\sim_1$ if $d\neq0$ and $q=1\bmod\hat G^\sim_1$ if $d=0$.
In the last subcase we can just simultaneously alternate the signs of~$d$ and~$q$.
Hence the exhaustive gauge is presented by $d\geqslant0$ and, if $d=0$, $q\geqslant0$.

\medskip

\noindent{\bf III.} Let $k=2$.
We choose a basis $\{(a_i,b_i,c_i,d_i,p_i),\ i=1,2\}$
of the space~$\mathcal V$ of tuples $(a,b,c,d,p)$ associated with~$A^{\max}$
and introduce the notation
\begin{gather*}
P_i=na_ix^2+b_ix+c_i,\quad Q_i=-3na_ix+d_i,\quad R_i=-(3n+m)a_ix+(n-m)p_i-2b_i.
\end{gather*}
Then $P_i\neq0,$ $i=1,2$, and the system for the functions~$f$ and~$h$ is written in the form
\begin{gather}\label{EqClassifyingK2}
\dfrac{f_x}f=\frac{Q_1}{P_1}=\frac{Q_2}{P_2},\quad\dfrac{h_x}h=\frac{R_1}{P_1}=\frac{R_2}{P_2}.
\end{gather}

Consider two different cases depending on
whether the first components of all elements from the space~$\mathcal V$ are zero.

If this is true, we can choose a basis of~$\mathcal V$ with $b_1=c_2=1$ and $b_2=c_1=0$.
Then system~\eqref{EqClassifyingK2} obviously implies that $f_x=h_x=0$.
As $fh\ne0$, up to $\hat G^{\sim}_1$-equivalence we obtain Case~3 of Table~1.

If the space~$\mathcal V$ contains tuples with nonzero first components,
we can assume without loss of generality that $a_1=1$ and $a_2=0$.
Then we have $b_2\ne0$ as otherwise system~\eqref{EqClassifyingK2} is not compatible.
We set $b_2=1$ and $b_1=0$ by changing the basis of~$\mathcal V$
and then set $c_2=0$ using a shift of~$x$ allowed by $\hat G^{\sim}_1$-equivalence.
Equations~\eqref{EqClassifyingK2} are consistent if and only if $c_1=0$.
They imply that $(f,h)=(x^{-3},\varepsilon |x|^{-3-\frac mn})\bmod\hat G^\sim_1,$ where $\varepsilon=\pm1$.
The corresponding maximal Lie invariance algebra
\[\langle\partial_t,\,nx^2\partial_x+x\partial_u,\,(m+n)t\partial_t+(m-n)x\partial_x-2\partial_u\rangle\]
is a two-dimensional extension of the kernel algebra~$A^\cap=\langle\p_t\rangle$.
This Lie symmetry extension is not included in Table~1 because it is reduced to Case~3
by the transformation~\eqref{EqEquivTransForExpFHtoConstCoeff} from the equivalence group $\hat G^\sim_1$.

\subsection{Case $\boldsymbol{m=0}$}\label{SectionOnLieSymmetriesCaseOfM0}

If $m=0$ then $\xi=nc_2x^2+c_1x+c_0$ and the classifying equations have the form
\begin{gather}\label{eq_lie_sym_m=0}
\begin{split}
&(nc_2x^2+c_1x+c_0)\frac{f_x}f=\tau_t-3nc_2x-\frac32c_1+n\eta^0,\\
&\eta^0_tf=(nc_2x^2+c_1x+c_0){h_x}+\left(3nc_2x+\frac32c_1-n\eta^0\right)h.
\end{split}
\end{gather}
Differentiating the latter equation with respect to $t$, we obtain the equation $f\eta^0_{tt}=-nh\eta^0_t$.
If $(h/f)_x\neq0$ then $\eta^0$ is a constant and hence the corresponding equation possesses
the same maximal Lie invariance algebra as in the case of general value of~$m$ for the same values of~$f$ and~$h$.
In what follows we assume that $(h/f)_x=0$,
i.e., $\varepsilon=h/f$ is a constant which equals $\pm1$ modulo $\hat G^\sim_1$.
Substituting the condition $h=\varepsilon f$ into~\eqref{eq_lie_sym_m=0},
we obtain that
\[\tau=c_4e^{-\varepsilon n t}+c_5,\quad\xi=nc_2x^2+c_1x+c_0,\quad\eta=c_2x+\varepsilon c_4e^{-\varepsilon n t}+c_3,\]
and the classifying equation on the function $f$ is the following one
\begin{gather*}
(nc_2x^2+c_1x+c_0)\dfrac{f_x}f=-3nc_2x-2c_1+nc_3.
\end{gather*}
It is obvious that this equation has the form~\eqref{eq_lie_sym_fc}.
Applying the method of furcate split in the same way as in Section~\ref{SectionOnLieSymmetriesGenCaseOfM},
we obtain only three $\hat G^\sim_1$-inequivalent Lie symmetry extensions, namely, Cases~4--6 of Table~1.
Case~5 is partitioned into the three $\hat G^\sim_1$-inequivalent subcases:
\begin{enumerate}\itemsep=0ex
\item
$f=|x|^d\colon\quad$
$\langle\partial_t,\,e^{-\varepsilon nt}(\partial_t+\varepsilon \partial_u),\,nx\partial_x+(d+2)\partial_u\rangle$;
\item
$f=e^x\colon\quad$
$\langle\partial_t,\,e^{-\varepsilon nt}(\partial_t+\varepsilon \partial_u),\,n\partial_x+\partial_u\rangle$;
\item
$f=(x^2+1)^{-\frac32}e^{d\arctan x}\colon\quad$
$\langle\partial_t,\,e^{-\varepsilon nt}(\partial_t+\varepsilon \partial_u),\,n(x^2+1)\partial_x+(x+d)\partial_u\rangle$.
\end{enumerate}
Here always $h=\varepsilon f$.
In the first subcase we should assume that $d\neq-3,0$ as otherwise this subcase is reduced to Case~6
with a wider Lie invariance algebra,
and $d\geqslant-3/2\bmod\hat G^\sim_1$, cf.~Lemma~\ref{LemmaOntransOfCoeffsOfClassifyingSystem}
and the transformation~\eqref{EqEquivTransForExpFHtoConstCoeff}.
In the last subcase we can just alternate the sign of~$d$.
Hence the exhaustive gauge up to $\hat G^\sim_1$-equivalence
is presented by $d\geqslant0$.

\subsection{Case $\boldsymbol{m=n}$}

If $m=n$ then $\eta^0$ is a constant and the classifying equations take the form
\begin{gather}\label{determ_eq_Lie}
\xi\frac{f_x}{f}=\tau_t-\frac32\xi_x+n\eta^0,\quad
\frac{1}{2n}\xi_{xxx}+\xi h_x+2\xi_xh=0.
\end{gather}
To carry out the group classification in this case,
it is necessary to use an additional gauge of the arbitrary elements of class~\eqref{eqRDfghExp_m=n=1}.
Consider both the gauges proposed in Section 2, namely, $f=1$ and $h=0$.

Choosing the gauge $f=1$ results in subclass~\eqref{eqRDfghExp_m=n_f=1}.
The first equation of~\eqref{determ_eq_Lie} implies that $\xi_{xx}=0$, i.e.,
$\xi=c_1x+c_2$, where $c_1$ and $c_2$ are constants and
$\eta=\frac1{2n}\xi_x+\eta^0$ is a constant, which is denoted by~$c_3$.
Then $\tau=(2c_1-nc_3)t+c_2$.
The classifying equation on the function $h$ can be written in the form
\[
(c_1x+c_2)h_x+ 2c_1h=0.
\]
It is easy to derive Cases~7, 8, 9 and~10 of Table~1.

Introducing the gauge $h=0$, we replace the study of class~\eqref{eqRDfghExp_m=n=1} by the study of class~\eqref{eqRDfghExp_m=n_h=0}.
Under this gauge, determining equations gives $\xi_{xxx}=0$ and hence $\tau$, $\xi$ and $\eta$ have the form~\eqref{eq_tau_xi_eta}.
The classifying equation on the function $f$ coincides with equation~\eqref{eq_lie_sym_f}.
The solution of~\eqref{eq_lie_sym_f} results in the following cases of Lie symmetries extensions:

\medskip

arbitrary $f$:\quad $\langle\partial_t,\,nt\partial_t-\partial_u\rangle$;

\medskip

$\displaystyle f=\exp\left({\int\frac{-3nax+d}{nax^2+bx+c}\,{\rm d}x}\right)$:\quad
$\langle\partial_t,\,nt\partial_t-\partial_u,\,(d+2b)t\partial_t+(nax^2+bx+c)\partial_x+ax\partial_u\rangle$;

\medskip

$f=1$:\quad $\langle\partial_t,\,nt\partial_t-\partial_u,\,\partial_x,\,2t\partial_t+x\partial_x\rangle$.

\medskip


The first and third cases obviously correspond to Cases~7 and~10 of Table~1, respectively.

Similarly to the general case of~$m$ and the case~$m=0$,
the second case of Lie symmetry extensions with $m=n$ under the gauge $h=0$
is partitioned into the following three $\hat G^\sim_1$-inequivalent subcases:
\begin{enumerate}\itemsep=0ex
\item $f=|x|^d\colon\quad$
$\langle\partial_t,\,nt\partial_t-\partial_u,\,(d+2)t\partial_t+x\partial_x\rangle$;

\item $f=e^x\colon\quad$
$\langle\partial_t,\,nt\partial_t-\partial_u,\,t\partial_t+\partial_x\rangle$;

\item $f=(x^2+1)^{-\frac32}e^{d\arctan x}\colon\quad$
$\langle\partial_t,\,nt\partial_t-\partial_u,\,ndt\partial_t+n(x^2+1)\partial_x+x\partial_u\rangle$.
\end{enumerate}
In the first subcase we should again assume that $d\neq-3,0$
as the value $d=0$ corresponds to the case of $f=1$ and $h=0$
with a wider Lie invariance algebra (Case~10 of Table~1)
and the value $d=-3$ is $\hat G^\sim_1$-equivalent to the value $d=0$,
cf.\ the consideration for $k=2$ in Section~\ref{SectionOnLieSymmetriesGenCaseOfM}.
Additionally, Lemma~\ref{LemmaOntransOfCoeffsOfClassifyingSystem} implies that
using the transformation~\eqref{EqEquivTransForExpFHtoConstCoeff} we can set the gauge
and $d\geqslant-3/2\bmod\hat G^\sim_1$.
In the last subcase the equivalence transformation alternating the sign of~$x$
leads to the gauge $d\geqslant0\bmod\hat G^\sim_1$.

\begin{remark}\label{RemarkOnTransBetweenClassificationListsInDifferentGauges}
The above subcases are related to cases of Table~1
via point transformations of the form~\eqref{EqGaugeTransformationForf1}
which belong to the equivalence group~$\hat G^{\sim}_{1,m=n}$.
Thus, the first subcase with $d\not\in\{-3,-3/2,0\}$, the second subcase and the third subcase with $d\neq0$
are mapped to Case~8 with
\[\alpha=\frac{d(d+3)}{4n(d+3/2)^2},\quad \alpha=\frac1{4n} \quad\mbox{and}\quad \alpha=\frac{d^2+9}{4nd^2},\]
respectively.
The first subcase with $d=-3/2$ and the third subcase with $d=0$ are mapped to Case~9 with
$\varepsilon=-1/(4n)$ and $\varepsilon=1/n$, respectively.
\end{remark}

\noprint{
\begin{gather*}
x^du_t=(e^{nu}u_x)_x\quad\longrightarrow\quad \tilde u_{\tilde t}=(e^{n\tilde u}\tilde u_{\tilde x})_{\tilde x}+
\alpha \tilde x^{-2}e^{n\tilde u},\quad \alpha=\tfrac {d(d+3)}{4n(d+3/2)^2};\\
e^xu_t=(e^{nu}u_x)_x\quad\longrightarrow\quad \tilde u_{\tilde t}=(e^{n\tilde u}\tilde u_{\tilde x})_{\tilde x}+
\alpha \tilde x^{-2}e^{n\tilde u},\quad \alpha=\tfrac1{4n};\\
(x^2+1)^{-\frac32}e^{d\arctan x}u_t=(e^{nu}u_x)_x\quad\longrightarrow\quad \tilde u_{\tilde t}=(e^{n\tilde u}\tilde u_{\tilde x})_{\tilde x}+
\alpha \tilde x^{-2}e^{n\tilde u},\quad \alpha=\tfrac{d^2+9}{4nd^2};\\
x^{-\frac32}u_t=(e^{nu}u_x)_x\quad\longrightarrow\quad \tilde u_{\tilde t}=(e^{n\tilde u}\tilde u_{\tilde x})_{\tilde x}+
\varepsilon e^{n\tilde u},\quad \varepsilon=-\tfrac {1}{4n};\\
(x^2+1)^{-\frac32}u_t=(e^{nu}u_x)_x\quad\longrightarrow\quad \tilde u_{\tilde t}=(e^{n\tilde u}\tilde u_{\tilde x})_{\tilde x}+
\varepsilon e^{n\tilde u},\quad \varepsilon=\tfrac{1}{n};\\
\end{gather*}
}

Analyzing the group classifications of the class~\eqref{eqRDfghExp_m=n=1} under both the gauges,
we conclude that each of the gauges $f=1$ and $h=0$ has certain advantages and disadvantages.
More precisely, under the gauge $h=0$ equations have a less number of summands and
the group classification list in the case $m=n$ is similar to the group classification lists in the cases of general~$m$ and $m=0$.
The gauge $h=0$ is also convenient in order to look for additional equivalence transformations
between the cases $m=0$ and $m=n$.
At the same time, the gauging transformation,
the corresponding equivalence group and the inequivalent values of the residuary arbitrary element, $h$,
arising in the course of group classification are much simpler for the gauge $f=1$.
In particular, it is very easy to partition into inequivalent cases under this gauge.
Moreover, the group classification list presented in Table~1, which involves the gauge $f=1$ in the case $m=n$, is well consistent with
the group classification list found for class~\eqref{eqRDfghTwoPower} in~\cite{VJPS2007}.
This is important for the study of contractions between cases of Lie symmetry extensions.

\subsection{Classification list and additional equivalence transformations}

Results of group classification for subclasses of class~\eqref{eqRDfghExp} are collected in Table~1.

\begin{table}[th!]\renewcommand{\arraystretch}{1.5}\footnotesize
\begin{center}
\textbf{Table 1.} Results of group classification of class~\eqref{eqRDfghExp} under the gauge $g=1$.
\\[2ex]
\begin{tabular}{|c|c|c|l|}
\hline
no.&$f(x)$&$h(x)$&\hfil Basis of $A^{\max}$ \\
\hline
\multicolumn{4}{|c|}{General case of $m$}\\
\hline
1&$\forall$&$\forall$&$\partial_t$\\
\hline
2&$f_1(x)$&$h_1(x)$&$\partial_t,\,(d+2b-pn)t\partial_t+
(nax^2+bx+c)\partial_x+(ax+p)\partial_u$\\
\hline
3&$1$&$\varepsilon$&$\partial_t,\,\partial_x,\,
2mt\partial_t+(m-n)x\partial_x-2\partial_u$\\
\hline\multicolumn{4}{|c|}{$m=0$, \quad $h\neq0$, \quad $(h/f)_x=0$}\\
\hline
4&$\forall$&$\varepsilon f(x)$&$\partial_t,\,e^{-\varepsilon nt}(\partial_t+\varepsilon \partial_u)$\\
\hline
5&$f_1(x)$&$\varepsilon f_1(x)$&$\partial_t,\,e^{-\varepsilon nt}(\partial_t+\varepsilon \partial_u),\,
n(nax^2+bx+c)\partial_x+(nax+2b+d)\partial_u$\\
\hline
6&$1$&$\varepsilon$&$\partial_t,\,e^{-\varepsilon nt}(\partial_t+\varepsilon \partial_u),\,\partial_x,\,nx\partial_x+2\partial_u$\\
\hline
\multicolumn{4}{|c|}{$m=n$ \quad or \quad $h=0$}\\
\hline
7&$\forall$&$\forall$&$\partial_t,\,nt\partial_t-\partial_u$\\
\hline
8&$1$&$\alpha x^{-2}$&$\partial_t,\,nt\partial_t-\partial_u,\,nx\partial_x+2\partial_u$\\
\hline
9&$1$&$\varepsilon$&$\partial_t,\,nt\partial_t-\partial_u,\,\partial_x$\\
\hline
10&1&0&$\partial_t,\,nt\partial_t-\partial_u,\,\partial_x,\,nx\partial_x+2\partial_u$\\
\hline
\end{tabular}
\end{center}
Here $n$, $\alpha$ and $\varepsilon$ are nonzero constants, $n=1\bmod\hat G^\sim_1$,
$\varepsilon=\pm1\bmod\hat G^\sim_1,$
\[f_1(x)=\exp\left({\int\frac{-3nax+d}{nax^2+bx+c}\,{\rm d}x}\right),\quad
h_1(x)=\varepsilon\exp\left({-\int\frac{(3n+m)ax+2b+(m-n)p}{nax^2+bx+c}\,{\rm d}x}\right),\]
and up to $\hat G^{\sim}_1$-equivalence
the parameter tuple~$(a,b,c,d,p)$ can be assumed to belong to the set
\[
\{(0,1,0,\bar d,(\bar q+2)/(n-m)),\ (0,0,1,1,\check p),\ (0,0,1,0,1),\ (1/n,0,1,\hat d,\hat p)\},
\]
where $(\bar d,\bar q)\neq(0,0),(-3,-3-m/n)$ and 
modulo $\hat G^{\sim}_1$ we can also set 
$\bar d\geqslant-3/2$ and, if $\bar d=-3/2$, $\bar q\geqslant-3/2-m/(2n)$;
$\hat d\geqslant0$ and, if $\hat d=0$, $\hat p\geqslant0$.
In Case~5 the parameter~$p$ should be neglected.
In Case~7 the arbitrary element~$f$ (resp.~$h$) can be additionally gauged
by transformations from~$\hat G^\sim_{1,m=n}$.
For example, we can set $f=1$.
\end{table}

There exist additional equivalence transformations between classification cases presented in Table~1.
Namely the point transformation
\begin{equation}\label{EqAddEquivTransOfEqRDfghPower}
t'=\frac 1{\varepsilon n}e^{\varepsilon n t},\quad
x'=x,\quad
u'=u-\varepsilon t
\end{equation}
links the equations
$
f(x)u_t=(g(x)e^{nu}u_x)_x+\varepsilon f(x)$
and
$f(x')u'_{t'}=(g(x')e^{nu'}u'_{x'})^{}_{x'}.
$
This transformation belongs to no equivalence group found in Section~2
and reduces Cases~4--6 of Table~1 to the set of cases `$m=n$ or $h=0$'.
For the reduction to be precise to Cases~7--10 of Table~1,
for Case~5 transformation~\eqref{EqAddEquivTransOfEqRDfghPower} should be composed with
an appropriate transformation of the form~\eqref{EqGaugeTransformationForf1}.
Such compositions map subcases of Case~5 to subcases of Case~8 and~9,
cf.\ Remark~\ref{RemarkOnTransBetweenClassificationListsInDifferentGauges}.
The transformations described exhaust additional equivalence transformations
within the classification list from Table~1.
It~is proved in the next section within the framework of admissible transformations.
Thus, transformation~\eqref{EqAddEquivTransOfEqRDfghPower}
can also be included in the framework of conditional equivalence but
the corresponding conditional equivalence group is too complicated.

As a result, we obtain the following assertion.

\begin{theorem}\label{TheoremOnClassificationOfRDfghExpEqsWRTPointTrans}
Up to point transformations, a complete list of Lie symmetry extensions for equations from class~\eqref{eqRDfghExp}
is exhausted by Cases~1--3 and~7--10 of Table~1.
\end{theorem}

Therefore, the maximal dimension of Lie symmetry algebras for equations from class~(\ref{eqRDfghExp}) equals four,
and for equations with $m\ne0,n$ this dimension equals three.

In the Table~1 we do not present the partitions of Cases~2 and~5 into inequivalent subclasses in detail.
The complete description of these partitions is given
in Sections~\ref{SectionOnLieSymmetriesGenCaseOfM} and~\ref{SectionOnLieSymmetriesCaseOfM0}.

\begin{corollary}
If an equation from class~(\ref{eqRDfghExp}) is invariant with respect to a four-dimensional Lie algebra
then it is reduced using point transformations to the equation $u_t=(e^uu_x)_x$.
\end{corollary}

\begin{corollary}
If an equation from class~(\ref{eqRDfghExp}) with $m\ne0,n$ possesses a three-dimensional Lie invariance algebra
then it is mapped by a point transformation to the equation $u_t=(e^uu_x)_x\pm e^{mu}$.
\end{corollary}

\section{Classification of admissible (form-preserving) transformations}\label{SectionClassificationOfAdmTrans}

In contrast to the construction of equivalence groups, in the course of the complete study of admissible transformations
we cannot vary arbitrary elements.
The description of the set of admissible transformations will be presented
in terms of conditional equivalence groups and normalized subclasses.

The consideration of Section~\ref{SectionOnPreliminaryStudyOfAdmTrans} implies
that any admissible transformation in class~\eqref{eqRDfghExp} is of the general form~\eqref{adm_tr},
where the parameters satisfy the inequality $T_t\varphi_x\delta_3\neq0$ and equations~\eqref{det_eq_adm_tr1} and~\eqref{det_eq_adm_tr2}.
Differentiating the first equation of~\eqref{det_eq_adm_tr1} with respect to $t$,
we have that $V=\theta(x)/T_t$. Recall that $V$ denotes the expression $e^{\frac{n}{\delta_3} \psi}$.
Then the determining equations become
\begin{gather}\label{det_eq_adm_tr1a}f\tilde g\theta-\tilde f g \varphi_x^2=0,\quad
2\frac{\theta_x}{\theta}-\frac{g_x}g+\frac{{\tilde g}_{\tilde x}}{\tilde g} \varphi_x  -\frac{\varphi_{xx}}{\varphi_x}=0,
\\\label{det_eq_adm_tr2a}
 \frac{1}{\varphi_x}\left(\tilde g\frac{\theta_x}{\varphi_x}\right)_xe^{nu}
-\tilde fT_t\left(\frac{1}{T_t}\right)_t-n\frac{\tilde f}fhe^{mu}+
\frac n{\delta_3}T_te^{\tilde m  \psi}\tilde h e^{\tilde m \delta_3 u}=0.
\end{gather}
Solving equations~\eqref{det_eq_adm_tr1a}, we express the new arbitrary elements~$\tilde f$ and~$\tilde g$ via the old ones:
\[
\tilde f=\frac{\delta_0}{\theta\varphi_x}\,f,\quad \tilde g=\frac{\delta_0\varphi_x}{\theta^2}\,g,
\]
where $\delta_0$ is an arbitrary nonzero constant.

The splitting of equation~\eqref{det_eq_adm_tr2a} with respect to~$u$
and the subsequent integration of the determining equations appreciably depend on values of~$m$ and $\tilde m$.
Consider possible cases for these values.

If $\tilde m=m/\delta_3\not=0$,
equation~\eqref{det_eq_adm_tr2a} implies
 $T_{tt}=0$, i.e.,  $T= \delta_1t+\delta_2$ and therefore $V=\theta(x)/\delta_1$.
 The other conditions obtained from~\eqref{det_eq_adm_tr2a} result in
\begin{gather*}\tilde m=\frac m{\delta_3}\not=0,\tilde n\colon\quad
\tilde h=\frac{\delta_0\delta_3}{\delta_1\theta\varphi_x}e^{-\frac m{\delta_3}\psi}h,\quad\mbox{where}\quad
\theta=\left(\delta_4\int \frac{dx}{g(x)}+\delta_5\right)^{-1};
\\[1.5ex]
\tilde m=\frac m{\delta_3}=\frac n{\delta_3}\colon\quad
\tilde h=\frac{\delta_0\delta_3}{n\theta\varphi_x}\left[\frac{nh}{\theta}-
\left(\frac{g\theta_x}{\theta^2}\right)_x\right],\quad\mbox{where}\quad
\theta\quad\mbox{is an arbitrary function.}
\end{gather*}
Therefore,  in the case $m\not=0,n$ (resp.\ the case $m=n$) any admissible transformation is induced by
a transformation from the generalized extended equivalence group~$\hat G^{\sim}$ (resp.\ the conditional equivalence group~$G^{\sim}_{m=n}$).

If $\tilde m=m=0$ then equation~\eqref{det_eq_adm_tr2a} implies an expression for~$\theta$ and
an equation connecting the ratios $\tilde h/\tilde f$ and $h/f$,
\[
\theta=\left(\delta_4\int \frac{dx}{g(x)}+\delta_5\right)^{-1},\quad
\frac{\tilde h}{\tilde f}=\frac{\delta_3}{T_t}\frac{h}{f}+\frac{\delta_3}n\left(\frac{1}{T_t}\right)_t.
\]
In view of the simplest differential consequence
\[
\biggl(\frac{\tilde h}{\tilde f}\biggr)_x=\frac{\delta_3}{T_t}\biggl(\frac{h}{f}\biggr)_x
\]
of the last equation we have $(\tilde h/\tilde f)_x\not=0$ and $T_{tt}=0$ if $(h/f)_x\not=0$.
Therefore, any admissible transformation for equations with $m=0$ and $(h/f)_x\not=0$ is induced by
a transformation from the equivalence group~$G^{\sim}_1$.

This is not the case if $(h/f)_x=0$.
Then $(\tilde h/\tilde f)_x=0$, i.e., the condition $(h/f)_x=0$ is invariant with respect to~$G^{\sim}_1$.
We denote the constants $h/f$ and $\tilde h/\tilde f$ by~$\alpha$ and~$\tilde\alpha$, respectively.
Then we get the following equation for the function~$T(t)$:
\[
\left(\frac1{T_t}\right)_t=-n\alpha\frac1{T_t}+\tilde n\tilde\alpha,
\]
where $\tilde n=n/\delta_3.$
We integrate this equation and present the general solution in such a form that
continuous dependence of it on the parameters~$\alpha$ and~$\tilde\alpha$ is obvious:
\begin{gather}\label{EqRDfghCondEquivGrouphfconstTransOfT}
\begin{split}
&\alpha\tilde\alpha\not=0\colon \quad
\frac{e^{\tilde n\tilde\alpha T}-1}{\tilde n\tilde\alpha}=\delta_1\frac{e^{n\alpha t}-1}{n\alpha}+\delta_2,
\qquad
\alpha=0,\ \tilde\alpha\not=0\colon \quad \frac{e^{\tilde n\tilde\alpha T}-1}{\tilde n\tilde\alpha}=\delta_1t+\delta_2,
\\[.5ex]
&\alpha\not=0,\ \tilde\alpha=0\colon \quad T=\delta_1\frac{e^{n\alpha t}-1}{n\alpha}+\delta_2,
\qquad
\alpha=\tilde\alpha=0\colon \quad T=\delta_1t+\delta_2.
\end{split}
\end{gather}
Here $\delta_1\ne0$ as $T_t\ne0$.

\begin{theorem}\label{TheoremOnRDfghCondEquivGrouphfconst}
The generalized extended equivalence group $\hat G^{\sim}_{m=0,(h/f)_x=0}$ of
the subclass of class~\eqref{eqRDfghExp}, which is singled out by the conditions $m=0$ and $(h/f)_x=0$,
consists of the transformations
\[
\begin{array}{l}
\tilde t=T(t),\quad \tilde x=\varphi(x), \quad
\tilde u=\delta_3u+\psi(t,x), \\[1ex]
\tilde f=\dfrac{\delta_0}{\varphi_x}\,\Psi f,\quad
\tilde g=\delta_0\varphi_x\Psi^2g,\quad
\tilde n=\dfrac{n}{\delta_3},
\end{array}
\]
where the smooth function~$T=T(t)$ with $T_t\ne0$ is defined by~\eqref{EqRDfghCondEquivGrouphfconstTransOfT},
$\alpha=h/f$ and $\tilde\alpha=\tilde h/\tilde f$ are constants,
$\varphi$ is an arbitrary smooth function of~$x$ with $\varphi_x\not=0$,
the functions~$\psi$ and~$\Psi$ are defined by the formulas
\[
\psi(t,x)= -\dfrac{\delta_3}{n}\ln\left|T_t(t)\Psi(x)\right|,\quad
\Psi(x)=\delta_4\int \frac{dx}{g(x)}+\delta_5,
\]
$\delta_j,$ $j=0,\dots,5$, are arbitrary constants,
$\delta_0\delta_1\delta_3\not=0$ and $(\delta_4,\delta_5)\ne(0,0)$.
\end{theorem}

In contrast to $G^{\sim}_{m=n}$, we do not use $\hat G^{\sim}_{m=0,(h/f)_x=0}$
in the course of group classification of class~\eqref{eqRDfghExp2}
because the application of this conditional equivalence group does not have a crucial influence on classification,
and the corresponding system $m=0$, $(h/f)_x=0$ for arbitrary elements is less obvious.
At the same time, transformations from $\hat G^{\sim}_{m=0,(h/f)_x=0}$ play the role of
additional equivalence transformations after completing the classification (see the previous section).

If $h\tilde h\not=0$, $m=0$ and $\tilde m=\tilde n$ then the
determining equations yield
\[
\frac hf=\frac1n\frac{T_{tt}}{T_t}, \quad\mbox{i.e.,}\quad \left(\frac{h}{f}\right)_x=0,
\qquad
\tilde h=-\frac{\delta_0\delta_3}{\varphi_x\theta}\left(\frac{\theta_x}{\theta^2}g\right)_x.
\]
Hence both equations~\eqref{eqRDfghExp}
and~\eqref{eqRDfghExptilde} are mapped to the subclass `$h=0$' of
the class under consideration. These mappings are realized by the
transformations from the corresponding conditional equivalence groups. 
We can assume that their images coincide. 
Therefore, the admissible transformation is composition of a
transformation~$\mathcal T_1$ of~\eqref{eqRDfghExp} to the
equation $f(x)u_t=(g(x)e^{nu}u_x)_x$ and a transformation~$\mathcal T_2$ of
the equation $f(x)u_t=(g(x)e^{nu}u_x)_x$ to~\eqref{eqRDfghExptilde}. The
transformation~$\mathcal T_1$ belongs to
$\hat G^{\sim}_{m=0,(h/f)_x=0}$ and the transformation~$\mathcal T_2$
belongs to $G^{\sim}_{m=n}$, and we can choose $\varphi=x$
in~$\mathcal T_1$.

The case $h\tilde h\not=0$, $m=n$ and $\tilde m=0$ is considered in similar way.

All the possible cases are exhausted.
We summarize the investigation of admissible transformations
in class~\eqref{eqRDfghExp} in the following assertion.

\begin{theorem}\label{TheoremOnRDfghSetOfAdmTrans}
Let the equations
\[
f(x)u_t=(g(x)e^{nu}u_x)_x + h(x)e^{mu}
\quad\mbox{and}\quad
\tilde f(\tilde x)u_{t}=(\tilde g(\tilde x) e^{\tilde n\tilde u}u_x)_x+\tilde h(\tilde x)e^{\tilde m\tilde u}
\]
be connected via a point transformation~$\mathcal T$ in the variables~$t$, $x$ and~$u$.
Then
\[
\mbox{either}\quad \frac{\tilde m}{\tilde n}=\frac mn \quad\mbox{or}\quad (m,\tilde m)=(0,\tilde n) \quad\mbox{or}\quad (m,\tilde m)=(n,0).
\]
The transformation~$\mathcal T$ is induced by a transformation from

\medskip

a) $\hat G^{\sim}$ if either $m\not=0,n$ or $m=0$, $(h/f)_x\not=0$;

b) $G^{\sim}_{m=n}$ if $m=n$ and $\tilde m\ne0$, then also $\tilde m=\tilde n$;

c) $\hat G^{\sim}_{m=0,(h/f)_x=0}$ if $m=\tilde m=0$, $(h/f)_x=0$, then also $(\tilde h/\tilde f)_x=0$.

\medskip

If $m=0$ and $\tilde m=\tilde n$ then $(h/f)_x=0$ and the transformation~$\mathcal T$ is the composition of
two transformations, from $\hat G^{\sim}_{m=0,(h/f)_x=0}$ and $G^{\sim}_{m=n}$, with
the intermediate equation having $h=0$.

The case with $m=n$ and $\tilde m=0$ is similar to the previous one.
\end{theorem}

\begin{corollary}\label{CorollaryOnRDfghSetOfAdmTrans}
Class~\eqref{eqRDfghExp} is represented as the union of its three maximal normalized subclasses separated by the conditions
\[
(h\not=0, \ m\not=0,n)\quad\mbox{or}\quad (m=0,\ (h/f)_x\not=0);\qquad
m=0, \ (h/f)_x=0;\qquad
m=n.
\]
Only the latter two subclasses have a non-empty intersection, and the intersection being the normalized subclass `$h=0$'.
\end{corollary}

Recall that the class of differential equations is called normalized if any
admissible transformation in this class belongs to its equivalence group.
See~\cite{Popovych&Eshraghi2005} for strong definitions.

The subsets of equations appearing in Corollary~\ref{CorollaryOnRDfghSetOfAdmTrans} 
are really subclasses of class~\eqref{eqRDfghExp} 
since they are singled out from class~\eqref{eqRDfghExp} 
by usual systems of differential equations and/or inequalities with respect to arbitrary elements.  
This is not obvious only for the first subclass. 
The condition corresponding to it is in fact equivalent to a single inequality, 
$(m-n)h\big(m^2+((h/f)_x)^2\big)\ne0$.

\section{Contractions}\label{SectionOnContractions}

Examples of nontrivial limits between equations admitting Lie symmetry extensions are known for a long time.
For instance, in~\cite{Bluman&Reid&Kumei1988} equations with exponential nonlinearities were excluded
from the group classification list of nonlinear diffusion equations as a separate case
and were just considered as a {\it limiting case} of equations with power nonlinearities.
At the same time, it looks more convenient to include such cases to classification lists
and then indicate connections between different classification cases via limiting processes.
Using the analogy with theory of Lie algebras such connections are called {\it contractions}.
A theoretical background on contractions of differential equations, their Lie symmetry algebras and solutions
was first discussed in~\cite{Ivanova&Popovych&Sophocleous2006II}.

In this section we relate, via contractions, the group classification lists obtained for classes~\eqref{eqRDfghExp} and~\eqref{eqRDfghTwoPower} 
in the present paper and in~\cite{VJPS2007}, respectively. 
For convenience of the presentation, inequivalent cases of Lie symmetry extension in class~\eqref{eqRDfghTwoPower} 
are collected in Table~2. 
Necessary explanations on the equivalence involved and thorough gauges of involved parameters by equivalence transformations 
are given in the appendix. 
Then contractions are used to construct exact solutions of equations from class~\eqref{eqRDfghExp} 
using known solutions of equations from class~\eqref{eqRDfghTwoPower}.
In Section~\ref{SectionOnContractionOfCLs} we demonstrate three different ways 
of a similar consideration for conservation laws: 
in terms of contractions of associated characteristic or conserved vectors or divergence expressions themselves.

\begin{table}[th!]\footnotesize
\renewcommand{\arraystretch}{1.6}
\begin{center}
\textbf{Table 2.} Results of group classification of class~\eqref{eqRDfghTwoPower}  under the gauge $g=1$ \cite{VJPS2007}.
\\[2ex]
\begin{tabular}{|c|c|c|c|l|}
\hline
no.&$n$&$f(x)$&$h(x)$&\hfil Basis of $A^{\max}$ \\
\hline
\multicolumn{5}{|c|}{General case}\\
\hline 1&$\forall$&$\forall$& $\forall$&
$\partial_t$\\
\hline 2&$\forall$&$f_1(x)$& $h_1(x)$&
$\partial_t, (d+2b-pn)t\partial_t+((n+1)ax^2+bx+c)\partial_x+(ax+p)u\partial_u$\\
\hline
3&$\forall$&1&$\varepsilon$&$\partial_t, \partial_x, 2(1-m)t\partial_t+(1+n-m)x\partial_x+2u\partial_u$ \\
\hline
\multicolumn{5}{|c|}{$m=1$, \quad $h\neq0$, \quad $(h/f)_x=0$}\\
\hline
4&$\forall$&$\forall$&$\varepsilon f$&$\partial_t, e^{-\varepsilon nt}(\partial_t+\varepsilon u\partial_u)$ \\
\hline 
5&$\forall$&$f_1(x)$&$\varepsilon f$&$\partial_t,
e^{-\varepsilon nt}(\partial_t+\varepsilon u\partial_u), n((n+1)ax^2+bx+c)\partial_x+(nax+2b+d)u\partial_u$ \\
\hline 
6&$\neq -\frac 43$&1&$\varepsilon$&$\partial_t, \partial_x,
e^{-\varepsilon nt}(\partial_t+\varepsilon u\partial_u), nx\partial_x+2u\partial_u$ \\
\hline 
7&$-\frac 43$&1&$\varepsilon$&$\partial_t, \partial_x,
e^{\frac 43\varepsilon t}(\partial_t+\varepsilon u\partial_u),
-\frac 43 x\partial_x+2u\partial_u,
-\frac 13x^2\partial_x+xu\partial_u$ \\
\hline
\multicolumn{5}{|c|}{$m=n+1$ \quad or \quad $h=0$}\\
\hline
8&$\forall$&$\forall$&$\forall$&$\partial_t, nt\partial_t-u\partial_u$ \\
\hline 
9&$\neq -\frac 43$&1&$\alpha x^{-2}$&$\partial_t,
nt\partial_t-u\partial_u, 2t\partial_t+x\partial_x$  \\
\hline 
10&$\neq -\frac 43$&1&$\varepsilon$&$\partial_t,
nt\partial_t-u\partial_u, \partial_x$  \\
\hline
11&$\neq -\frac 43$&1&0&$\partial_t, \partial_x, nt\partial_t-u\partial_u, 2t\partial_t+x\partial_x$ \\
\hline 12&$-\frac 43$&$e^x$&$\alpha$&$\partial_t,
t\partial_t+\frac34u\partial_u, \partial_x-\frac34u\partial_u$  \\
\hline 13&$-\frac 43$&1&0&$\partial_t, \partial_x,
\frac 43 t\partial_t+u\partial_u, 2t\partial_t+x\partial_x, -\frac{1}{3}x^2\partial_x+xu\partial_u$ \\
\hline
\end{tabular}
\end{center}
Here  $\alpha$ is arbitrary constant, $\alpha\not=0$ in Case~9, $\varepsilon=\pm1$,
\[
f_1(x)={\rm exp}\left [\int \frac{-(3n+4)ax+d}{(n+1)ax^2+bx+c}dx\right ],\quad
h_1(x)={\rm exp}\left [\int \frac{-(3(n+1)+m)ax+(n-m+1)p-2b}{(n+1)ax^2+bx+c}dx\right ],
\]
and it can be assumed up to $G^{\sim}_{g=1}$-equivalence  (see Theorem~\ref{TheoremOnEquivGroupg1}) 
that, if $n\neq-1$,
the parameter tuple~$(a,b,c,d,p)$ takes only the following inequivalent values:
\[
\{(0,1,0,\bar d,(\bar q+2)/(n+m-1)),\ (0,0,1,1,\check p),\ (0,0,1,0,1),\ (1/(n+1),0,1,\hat d,\hat p)\},
\]
where $\check p$ is an arbitrary constant; 
$\hat d\geqslant0$ and, if $\hat d=0$, $\hat p\geqslant0$;
\[
(\bar d,\bar q)\ne(0,0),\left(-\frac{3n+4}{n+1},-3-\frac m{n+1}\right);\quad
\bar d\geqslant-\frac{3n+4}{2(n+1)}\ \mbox{and, if\ } 
\bar d=-\frac{3n+4}{2(n+1)},\  \bar q\geqslant-\frac32-\frac m{2(n+1)}. 
\]
If $n=-1$, up to $G^{\sim}_{g=1}$-equivalence
the parameter tuple~$(a,b,c,d,p)$ can be assumed to belong to the set
\[
\{(0,1,0,d',p'),\ (0,0,1,0,1),\ (\varepsilon,0,1,0,p'')\},
\]
where
$d'$  and $p'$ are arbitrary constants, $p''\geqslant0$.
In Case~5 the parameter~$p$ should be neglected.
In Case~8 the parameter-functions~$f$ and~$h$ can be additionally
gauged with equivalence transformations from $G^{\sim}_{g=1,m=n+1}$ (see Theorem~\ref{CorollaryOnEquivGroupMN1}).
For example, we can put $f=1$ if $n\not=-4/3$ and $f=e^x$
otherwise.
\end{table}

\subsection{Contractions of equations and of Lie invariance algebras}

At first we apply the equivalence transformation
\begin{equation}\label{contraction}
\tilde t ={\delta} t,\quad \tilde x={\sqrt\delta}x,\quad
\tilde u=\delta(u-1),\quad\tilde n=\frac n{\delta},\quad \tilde m=\frac m{\delta}
\end{equation}
parameterized by a positive constant parameter $\delta$ to the equation from class~\eqref{eqRDfghTwoPower} 
with the values arbitrary elements $g=1$ and $f$ and $h$ presented in Case~2 of Table 2.
The constant parameters $a$, $b$, $c$, $d$ and $p$ are transformed in the following way
\begin{equation}\label{contraction_parameters}
\tilde a=a,\quad \tilde b=\frac b{\sqrt\delta},\quad \tilde c=c,\quad\tilde d=\frac d{\sqrt\delta},\quad\tilde p={\sqrt\delta}p,
\quad\tilde\alpha=\delta\alpha
\end{equation}
wherever this is relevant, i.e., we change parameters if and only if they appear in the values of arbitrary elements of the initial equation.
Then, we take the imaged equation and proceed to the limit $\delta\rightarrow +\infty$.
This results in the equation from class~\eqref{eqRDfghExp2} with the values of the arbitrary elements $f$ and $h$ presented in
Case~2 of Table 1. The same procedure establishes a contraction between the associated Lie algebras of vector fields.
The corresponding notation will be $2.2\rightarrow1.2$, where the first numbers indicate the numbers of the tables and the second numbers indicate the numbers of cases within the tables.
We present the complete list of contractions which replace power nonlinearities by exponential ones and, therefore,
connect cases of Lie symmetry extensions for classes~\eqref{eqRDfghExp} and~\eqref{eqRDfghTwoPower}:
\begin{gather*}2.1\rightarrow1.1,\quad2.2\rightarrow1.2,
\quad2.3\rightarrow1.3,\quad2.4\rightarrow1.4,\quad2.5\rightarrow1.5,\\
2.6\rightarrow1.6,\quad2.8\rightarrow1.7,\quad2.9\rightarrow1.8,\quad2.10\rightarrow1.9,\quad2.11\rightarrow1.10.
\end{gather*}

\subsection{Contractions of Lie reductions and exact solutions}

In~\cite{VJPS2007} we carried out Lie reductions and constructed Lie exact solution for equations from class~\eqref{eqRDfghTwoPower} with the values of arbitrary elements presented in Cases~9 and 12 of Table 2, which admit three-dimensional Lie symmetry algebras.
It is shown in the previous subsection that there is a contraction of Case~9 of Table 2 to Case~8 of Table~1.
The corresponding equations from classes~\eqref{eqRDfghExp} and~\eqref{eqRDfghTwoPower} are
\begin{gather}
\label{eq_exp_case8}
\tilde u_{\tilde t}=\left(e^{\tilde n\tilde u}\tilde u_{\tilde x}\right)_{\tilde x}+\tilde\alpha\tilde x^{-2}e^{\tilde n\tilde u},
\\ \label{eq_power_case9}
u_t=(u^nu_x)_x+\alpha x^{-2}u^{n+1},
\end{gather}
whose maximal Lie invariance algebras~$\tilde{\mathfrak g}$ and~$\mathfrak g$ are generated by the vector fields
$\tilde X_1=\p_{\tilde t}$, $\tilde X_2=n\tilde t\p_{\tilde t}-\p_{\tilde u}$,
$\tilde X_3=n\tilde x\p_{\tilde x}+2\p_{\tilde u}$ and $X_1=\p_t$, $ X_2=nt\p_t-u\p_u$,
$ X_3=nx\p_x+2u\p_u$,
respectively.
The contraction $2.9\rightarrow1.8$ can be realized using the  simpler transformation
\begin{equation}\label{contraction1} \tilde t =t,\quad \tilde x= x,\quad
\tilde u=\delta(u-1),\quad\tilde n=\frac n{\delta},\quad \tilde \alpha={\delta}\alpha
\end{equation}
than transformation~\eqref{contraction}.
In the course of this contraction the algebra~$\mathfrak g$
is contracted to the algebra~$\tilde{\mathfrak g}$
as a Lie algebra of vector fields in the space of~$(t,x,u)$.
Namely, $X_1\rightarrow\tilde X_1$, $X_2\rightarrow\tilde X_2$ and $X_3\rightarrow\tilde X_3$.

Let us study the related contractions of Lie reductions of equation~\eqref{eq_power_case9}
to ones of equation~\eqref{eq_exp_case8}.
Inequivalent Lie reductions of these equations with respect to one-dimensional subalgebras
of the corresponding maximal Lie invariance algebras are exhausted by those presented in Tables~3 and~4.
For convenience we omit tildes in Table~4.
The transformations of the invariant independent and dependent variables,
which are induced by transformation~\eqref{contraction1},
take the form $\tilde \varphi=\delta(\varphi-1)$ and $\tilde\omega=\omega$
in all cases of Table~3.

\begin{table}[th!]\begin{center}\footnotesize\renewcommand{\arraystretch}{1.6}
\textbf{Table~3.} Lie reductions for Case~9 of Table~2.
\\[2ex]
\begin{tabular}{|r|c|c|c|c|}\hline
N\hfil\null& $X$ & $\omega$ & $u =$ & Reduced ODE \\
\hline
1 & $X_2 - \mu X_3$ & $x|t|^\mu$ &
$|t|^{-\frac{1+2\mu}{n}}\varphi(\omega)$
& $(\varphi^{n}\varphi_\omega)_\omega -\mu\omega\varepsilon\varphi_\omega+\frac{1+ 2\mu}{n}\varepsilon\varphi+
\alpha\omega^{-2}\varphi^{n+1}=0$, $\varepsilon=\sign t$
\\
2& $X_3$ & $t$ & $|x|^\frac 2n \varphi(\omega)$ &
 $n^2\varphi_\omega-(\alpha n^2+2n+4)\varphi^{n+1}=0\quad$
\\
3& $X_3\pm X_1$ & $xe^{\mp t}$ & $e^{\pm\frac {2t}n}\varphi(\omega)$ & 
$(\varphi^{n}\varphi_\omega)_\omega \pm\omega\varphi_\omega\mp\frac 2n\varphi+\alpha\omega^{-2}\varphi^{n+1}=0$
\\
4& $X_1$&$x$&$\varphi(\omega)$&
$(\varphi^{n}\varphi_\omega)_\omega+\alpha\omega^{-2}\varphi^{n+1}=0$
\\
\hline
\end{tabular}
\end{center}
\end{table}

\begin{table}[th!]\begin{center}\footnotesize\renewcommand{\arraystretch}{1.6}
\textbf{Table~4.} Lie reductions for Case~8 of Table~1.
\\[2ex]
\begin{tabular}{|r|c|c|c|c|}\hline
N\hfil\null& $X$ & $\omega$ & $u =$ & Reduced ODE \\
\hline
1 & $X_2 - \mu X_3$ & $x|t|^\mu$ &
$\varphi(\omega)-\frac{1+2\mu}n\ln|t|$
& $(e^{n\varphi})_{ \omega\omega}-\mu n\varepsilon\omega\varphi_\omega+(1+2\mu)\varepsilon+\alpha  n\omega^{-2}e^{n\varphi}=0$, 
$\varepsilon=\sign t$
\\
2& $X_3$ & $t$ & $\varphi(\omega)+\frac2n\ln|x|$ &
 $n\varphi_\omega-(\alpha n+2)e^{n\varphi}=0$
\\
3& $X_3\pm X_1$ & $xe^{\mp t}$ & $\varphi(\omega)\pm\frac2{n}t$ & $(e^{n\varphi})_{\omega\omega}\pm n\omega\varphi_\omega
\mp2+
 \alpha  n\omega^{-2}e^{  n \varphi}=0$
\\
4& $X_1$&$x$&$\varphi(\omega)$&
$(e^{n\varphi})_{\omega\omega}+\alpha n\omega^{-2}e^{n\varphi}=0$
\\
\hline
\end{tabular}
\end{center}
\end{table}

Consider Case~1 of Table~3 in detail.
The transformed version and the corresponding limit of the ansatz $u=|t|^{-\frac{1+2\mu}{n}}\varphi(\omega)$ are
\[
\left(1+\frac{\tilde u}{\delta}\right)^{\delta\tilde n}=
|\tilde t|^{-(1+2\mu)}\left(1+\frac{\tilde \varphi}{\delta}\right)^{\delta\tilde n}
\quad\rightarrow\quad 
e^{\tilde n\tilde u}=|\tilde t|^{-(1+2\mu)}e^{\tilde n\tilde \varphi}\quad\mbox{at}\quad
\delta \rightarrow +\infty.
\]
Therefore, the contracted ansatz is $\tilde u =\tilde\varphi(\tilde\omega) -\frac{1+2\mu}{\tilde n} \ln |t|$.
The reduced equation from Case~1 of Table~3 is mapped by transformation~\eqref{contraction1}
to the equation
\[
\dfrac{\delta \tilde n}{\delta \tilde n+1}
\left[\left(1+\frac{\tilde \varphi}{\delta}\right)^{\delta\tilde n+1}\right]_{\tilde\omega\tilde\omega} 
-\mu\tilde n\varepsilon\tilde\omega{\tilde\varphi}_{\tilde\omega}
+(1+ 2\mu)\varepsilon\left(1+\frac{\tilde \varphi}{\delta}\right)+
\frac{\tilde\alpha\tilde n}{\tilde\omega^{2}}\left(1+\frac{\tilde \varphi}{\delta}\right)^{\delta\tilde n+1}
=0.
\]
Then the limit process at $\delta \rightarrow +\infty$ leads to the equation
\[
\left(e^{\tilde n\tilde\varphi}\right)_{\tilde\omega\tilde\omega} 
-\mu\tilde n\varepsilon\tilde\omega{\tilde\varphi}_{\tilde\omega}
 +(1+2\mu)\varepsilon+
\frac{\tilde\alpha\tilde n}{\tilde\omega^{2}}e^{\tilde n\tilde\varphi}=0
\]
which is also obtained from equation~\eqref{eq_exp_case8}
by the reduction with respect to the contracted ansatz and presented by Case 1 of Table~4.

Analogously we obtain contractions of reductions $3.2\rightarrow4.2$, $3.3\rightarrow4.3$ and $3.4\rightarrow4.4$.

For Cases~2 and~4 of Table~3 exact solutions of reduced equations were found in~\cite{VJPS2007}.
The substitution of these solutions to the respective ansatzes results in
the following exact solutions of  equation~\eqref{eq_power_case9}:
\begin{gather}\label{sol1}
u=\left|\frac{x^2}{C-(\alpha n+2+4n^{-1})t}\right|^{\frac1n},
\\[1.5ex] \label{sol2}
u=\begin{cases}
\big|C_1\sqrt x\ln x+C_2\sqrt x\big|^\frac1{n+1},
&\text{if}\quad \alpha'=0,\\[.5ex]
\big|C_1x^{\varkappa_1}+C_2x^{\varkappa_2}\big|^\frac1{n+1},
&\text{if}\quad \alpha'>0,\\[.5ex]
\big|C_1\sqrt x\sin(\sigma\ln x)+C_2\sqrt x\cos(\sigma\ln x)\big|^\frac 1{n+1},
&\text{if}\quad \alpha'<0,
\end{cases}
\end{gather}
where
\[
\alpha'=1-4\alpha(n+1),
\quad\varkappa_{1,2}=\frac{1\pm\sqrt{\alpha'}}2,
\quad \sigma=\frac{\sqrt{-\alpha'}}{2}.
\]
Here and in what follows $C$, $C_1$ and $C_2$ are arbitrary constants.
Applying transformation~\eqref{contraction1} to solution~\eqref{sol1} and proceeding with
the limit $\delta\rightarrow+\infty$, we obtain
\begin{gather*}
\left(1+\frac{\tilde u}\delta\right)^{\delta\tilde n}=
\tilde x^2\left(C-\Bigl(\tilde\alpha\tilde n+2+\frac4{\tilde n\delta}\Bigr)\tilde t\right)^{-1}
\quad\rightarrow\quad
e^{\tilde n\tilde u}={\tilde x}^2\left(C-\left(\tilde\alpha\tilde n+2\right)\tilde t\right)^{-1}.
\end{gather*}
As a result, we construct the exact solution
\begin{gather*}
\tilde u=
\frac1{\tilde n}\ln\left|\frac{{\tilde x}^2}{C-\left(\tilde\alpha\tilde n+2\right)\tilde t}\right|.
\end{gather*}
for equation~\eqref{eq_exp_case8}.
Applying the same technique to solutions~\eqref{sol2} leads to the steady-state exact solutions of~\eqref{eq_exp_case8}:
\begin{gather*}
\tilde u=\begin{cases}
\frac1{\tilde n}\ln\big|C_1\sqrt{\tilde x}\ln\tilde x+C_2\sqrt{\tilde x}\big|,
&\text{if}\quad \tilde\alpha'=0,\\[.5ex]
\frac1{\tilde n}\ln\big|C_1\tilde x^{\varkappa_1}+C_2\tilde x^{\varkappa_2}\big|,
&\text{if}\quad \tilde\alpha'>0,\\[.5ex]
\frac1{\tilde n}\ln\big|C_1\sqrt{\tilde x}\sin(\sigma\ln{\tilde x})+C_2\sqrt{\tilde x}\cos(\sigma\ln{\tilde x})\big|,
&\text{if}\quad \tilde\alpha'<0,
\end{cases}
\end{gather*}
where
\[
\quad \tilde\alpha'=1-4\tilde\alpha\tilde n,
\quad\varkappa_{1,2}=\frac{1\pm\sqrt{\tilde\alpha'}}2,
\quad \sigma=\frac{\sqrt{-\tilde\alpha'}}{2}.
\]
Another way for finding this solution is to integrate the reduced equation of Case~4 from Table~4.
By the obvious transformation $\hat\varphi=e^{n\varphi}$ the reduced equation is mapped
to the Euler equation $\omega^2\hat\varphi_{\omega\omega}+\alpha n\hat\varphi=0$.

\subsection{Contractions of conservation laws}\label{SectionOnContractionOfCLs}

We use contractions in order to construct conservation laws of equations
from class~\eqref{eqRDfghExp} with $g=1$ using results obtained in~\cite{VJPS2007}
for equations from class~\eqref{eqRDfghTwoPower} with the same gauge of~$g$. 
Note that the consideration can be easily extended 
to the entire classes~\eqref{eqRDfghExp} and~\eqref{eqRDfghTwoPower} 
using transformations from the corresponding equivalence groups.

Roughly speaking, a conservation law of a system~$\mathcal L$ of differential equations is a divergence expression that
vanishes on solutions of this system. 
Thus, in the case of two independent variables~$t$ and~$x$ and one unknown function~$u$ 
the general form of conservation laws is
$D_tF(t,x,u_{(r)})+D_xG(t,x,u_{(r)})=0$ whenever $u$ is a solution of~$\mathcal L$. 
Here $D_t$ and $D_x$ are the operators of total differentiation
with respect to $t$ and $x$, respectively, and $u_{(r)}$ denotes the set of all the
derivatives of the functions $u$ with respect to $t$ and $x$ of order not
greater than~$r$, including $u$ as the derivative of the zero order.
The components $F$ and $G$ of the
conserved vector~$(F,G)$ are called the {\em density} and the {\em
flux} of the conservation law. Two conserved vectors $(F,G)$ and
$(F',G')$ are equivalent if there exist such functions~$\hat F$,
$\hat G$ and~$H$ of~$t$, $x$ and derivatives of~$u$ that $\hat F$
and $\hat G$ vanish for all solutions of~$\mathcal{L}$~and
$F'=F+\hat F+D_xH$, $G'=G+\hat G-D_tH$. 
A conserved vector is called trivial if it is equivalent to the zero conserved vector. 

It is found in~\cite{VJPS2007} that there are only two subclasses of equations of form~\eqref{eqRDfghTwoPower}
which admit nontrivial conserved vectors.
Thus, assuming the gauge $g=1$,
each equation from class~\eqref{eqRDfghTwoPower} with $m=n+1$, i.e., an equation of the form
$
f(x)u_t=(u^nu_x)_x+h(x)u^{n+1},
$
admits two linearly independent conservation laws 
with the following conserved vectors $(F^i,G^i)$ and the characteristics $\lambda^i,$ $i=1,2$:
\begin{gather*}
n\neq-1\colon\quad \Bigl(\, \varphi^ifu,\,
-\varphi^iu^nu_x+\varphi^i_x\tfrac{u^{n+1}}{n+1}\,\Bigr),\quad
\lambda^i=\varphi^i,\ i=1,2;\\
n=-1\colon\quad (\, xfu,\ -xu^{-1}u_x+\ln u\,),\quad\lambda^1= x;\qquad (\, fu,\
-u^{-1}u_x\,),\quad\lambda^2=1.
\end{gather*}
Here $\beta_1$ and $\beta_2$ are arbitrary constants.
The functions $\varphi^i=\varphi^i(x)$, $i=1,2$, form a
fundamental set of solutions of the second-order linear
ordinary differential equation
$
\varphi_{xx}+(n+1)h\varphi=0.
$

The other subclass of equations admitting nontrivial conserved vectors is singled out from~\eqref{eqRDfghTwoPower}
under the gauge~$g=1$ by the conditions $m=1$ and $h=\alpha f$, where $\alpha$ is an arbitrary constant, i.e., it consists of equations of the form
\begin{equation}\label{eq_cons_laws}
f(x)u_t=(u^nu_x)_x+\alpha f(x)u.
\end{equation}
The corresponding conserved vectors and characteristics are
\begin{gather}\label{CL}
n\neq-1 \colon\quad\begin{array}{l}\bigl(\,xe^{-\alpha t}fu,\
e^{-\alpha t}\bigl(-xu^nu_x+\tfrac{u^{n+1}}{n+1}\bigr)\,\bigr),\quad
\lambda^1=xe^{-\alpha t},\\ [.5ex]
(\,e^{-\alpha t}fu,\ - e^{-\alpha t}u^nu_x\,),\quad \lambda^2=e^{-\alpha t};
\end{array}\\\nonumber
n=-1\colon\quad\begin{array}{l} (\,xe^{-\alpha t}fu,\ e^{-\alpha t}(-xu^{-1} u_x+\ln
u)\,),\quad \lambda^1=xe^{-\alpha t}, \\[1ex] (\,e^{-\alpha t}fu,\ -e^{-\alpha
t}u^{-1}u_x\,),\quad \lambda^2=e^{-\alpha t}.
\end{array}
\end{gather}

In order to contract equations from class~\eqref{eqRDfghTwoPower} to equations from class~\eqref{eqRDfghExp},
we should vary the arbitrary element~$n$.
This is why only the case of general $n$ is appropriate for contractions.
There are three different ways in order to realize contractions of conservation laws.
Namely, we can contract characteristics of conservation laws,
their conserved vectors or conservation laws as divergent expressions themselves.

We illustrate these possibilities in detail using equations~\eqref{eq_cons_laws} with $n\neq-1$
and their conservation laws associated the same characteristic $\lambda^1=xe^{-\alpha t}$. 
The corresponding conserved vectors are presented in~\eqref{CL}.
At first we apply equivalence transformation~\eqref{contraction1} to
equation~\eqref{eq_cons_laws} and proceed to the limit $\delta\rightarrow+\infty$.
As a result, we obtain the class of  equations (tildes are omitted)
\begin{equation}\label{eq_cons_laws_exp}
f(x)u_t=(e^{nu}u_x)_x+\alpha f(x),
\end{equation}
i.e., equations from class~\eqref{eqRDfghExp2} with $m=0$ and $h=\alpha f$.

For the image $\tilde\lambda^1$ of the characteristic $\lambda^1=xe^{-\alpha t}$
with respect to transformation~\eqref{contraction1} we have that
$\tilde\lambda^1\,\rightarrow\, x$ if $\delta\rightarrow+\infty$.
Now we are able to construct the corresponding conservation law of~\eqref{eq_cons_laws_exp}
using the characteristic obtained as an integrating factor.
After the multiplication by $x$ the equation~\eqref{eq_cons_laws_exp}
can be written in divergent form as
\begin{equation}\label{cons_law_axample}
D_t\left(x f u-\alpha xf t\right)+D_x\left(-xe^{nu}u_x+
\frac1n e^{nu}\right)=0.
\end{equation}
Therefore, we construct conservation law~\eqref{cons_law_axample} of equation~\eqref{eq_cons_laws_exp}
via carrying out a limiting process of characteristics.
Another way is to directly deal with divergence expressions.
Thus the conservation law
\[
D_t\left(xe^{-\alpha t}fu\right)+
D_x\left( -e^{-\alpha t}xu^nu_x+e^{-\alpha t}\frac{u^{n+1}}{n+1}\right)=0
\]
of~\eqref{eq_cons_laws} with $n\neq-1$
is transformed by~\eqref{contraction1} to
\[
D_{\tilde t}\left(\tilde x e^{-\frac{\tilde\alpha}{\delta} \tilde t}f\left(\frac{\tilde u}{\delta}+1\right)\right)+
D_{\tilde x}\left(
-e^{-\frac{\tilde\alpha}{\delta}\tilde t}\tilde x\left(\frac{\tilde u}{\delta}+1\right)^{\tilde n\delta}\frac{{\tilde u}_{\tilde x}}{\delta}
+e^{-\frac{\tilde\alpha}{\delta}\tilde t}\frac1{\tilde n\delta+1}{\left(\frac{\tilde u}{\delta}+1\right)^{\tilde n\delta+1}}\right)=0.
\]
Multiplying the obtained expression by $\delta$ and adding the term $-\tilde xf\delta$ under $D_{\tilde t}$,
we proceed to the limit $\delta\rightarrow+\infty$.
(As the term $-\tilde xf\delta$ depends only on $\tilde x$, it is negligible in view of the differentiation with respect to~$\tilde t$.)
Omitting tildes, we exactly obtain the conservation law~\eqref{cons_law_axample}.

Conserved vectors can be contracted in the same way using the fact that we can add expression which
depends on $\tilde x$ only to the density component up to equivalence of conserved vectors.

Contracting  the conservation laws obtained in~\cite{VJPS2007} jointly with the corresponding equations,
we obtain the following assertion.

\begin{theorem}\label{TheoremOnClassificationCLsOfEqRDf1h}
A complete list of equations from class~\eqref{eqRDfghExp2} possessing
nontrivial conservation laws is exhausted by the following ones.
\setcounter{mcasenum}{0}

\vspace{1ex}

$\makebox[6mm][l]{\refstepcounter{mcasenum}\themcasenum\label{CaseCLsRDfghBinA}.}
m=n:$
\nopagebreak\\[1ex]\hspace*{3\parindent}%
$\bigl(\, \varphi^ifu,\,
-\varphi^ie^{nu}u_x+\frac{1}{n}\varphi^i_xe^{nu}\,\bigr),\
\varphi^i,\ i=1,2.$

$\makebox[6mm][l]{\refstepcounter{mcasenum}\themcasenum\label{CaseCLsRDfghBinA1}.}
m=0,\quad h=\alpha f:$
\nopagebreak\\[1ex]\hspace*{3\parindent}%
$\bigl(\,xf(u-\alpha t),\ - x
e^{nu}u_x+\frac{1}{n}e^{nu}\,\bigr),\ x;$ \
$(\,f(u-\alpha t),\ -e^{nu}u_x\,),\ 1.$

\vspace{1ex}

\noindent Here the functions $\varphi^i=\varphi^i(x)$, $i=1,2$, form a
fundamental set of solutions of the second-order linear
ordinary differential equation $\varphi_{xx}+nh\varphi=0$ and $\alpha$ is an arbitrary constant.
\end{theorem}

Simultaneously with constraints on the arbitrary elements
we also present conserved vectors and characteristics of the basis
elements of the corresponding space of conservation laws.

\section{Conclusion}
In this paper we carry out the extended Lie symmetry analysis
of equations from class~\eqref{equation_RD_general}, where $A$ and $B$ are
exponential functions of $u$. In this sense the present paper continues the series of papers~\cite{VJPS2007,
Vaneeva&Popovych&Sophocleous2009,vaneeva_proc,Vaneeva&Popovych&Sophocleous2011}, where
equations~\eqref{equation_RD_general} with power nonlinearities and semilinear equations of this form
with exponential source were studied.

Lie symmetries, admissible transformations and conservation laws of equations from class~\eqref{eqRDfghExp} 
are exhaustively classified.
This difficult task is achieved due to the usage of generalized equivalence groups instead of usual ones, 
nontrivial gauging of arbitrary elements 
and the representation of class~\eqref{eqRDfghExp} as a union of normalized subclasses.
Moreover, in the course of group classification of equations with $m=n$
the associated conditional equivalence group is used,
which essentially simplified calculations and the final result.

We also study limit processes, called contractions,
between equations from classes~\eqref{eqRDfghTwoPower} and~\eqref{eqRDfghExp}
jointly with limit processes between the corresponding Lie invariance algebras.
This allows us to predict and then to check the results of group classification for class~\eqref{eqRDfghExp}.
Moreover, we construct conservation laws and exact solutions for equations from class~\eqref{eqRDfghExp}
using contractions and the related results obtained for equations of the form~\eqref{eqRDfghTwoPower} 
in~\cite{VJPS2007}.

The feature of this paper in comparison with the preceding ones is that the method of furcate split
applied to classify Lie symmetries classification as well as the usage of the direct method 
for finding equivalence and admissible transformations are described in detail.

Existing examples on group classification of variable--coefficient
diffusion--convection equations~\cite{Ivanova&Sophocleous2006,Ivanova&Popovych&Sophocleous2006I,Popovych&Ivanova2004NVCDCEs}
show that the subclasses with power or exponential nonlinearities $A$ and $B$ are usually the most complicated
to be classified.
In a forthcoming paper we intend to complete Lie symmetry analysis of equations from class~\eqref{equation_RD_general},
where the nonlinearities $A$ or $B$ are not power or exponential functions.

\appendix

\section{Gauging of parameters in classification list\\ with power nonlinearities}

We briefly describe results on the group classification of class~\eqref{eqRDfghTwoPower} 
obtained in~\cite{VJPS2007}, which are needed for the proper understanding of 
the final list of inequivalent Lie symmetry extensions collected in Table~2. 
In order to attain to the complete similarity to the presentation of the group classification 
of class~\eqref{eqRDfghExp} in the present paper, we carry out certain modifications and enhancements 
of results from~\cite{VJPS2007}. 
In particular, we thoroughly gauge parameters appearing in Cases~2 and~5 of Table~2 and 
partition these cases into simpler inequivalent subcases in the same way as this is done 
for class~\eqref{eqRDfghExp} in Section~\ref{SectionOnLieSymmetries}.

Via gauging of the arbitrary element~$g$ by transformations from the corresponding equivalence group,
the group classification of class~\eqref{eqRDfghTwoPower} is reduced to that of its subclass 
singled out by the constraint $g=1$. 

\begin{theorem}\label{TheoremOnEquivGroupg1}
The generalized equivalence group~$G^{\sim}_{g=1}$ of the subclass of equations
having the form~\eqref{eqRDfghTwoPower} with $g=1$ consists of the transformations
\[
\begin{array}{l}
\tilde t=\delta_1 t+\delta_2,\quad \tilde x=\dfrac{\delta_6
x+\delta_7}{\delta_4 x+\delta_5}, \quad
\tilde u=\delta_3|\delta_4 x+\delta_5|^{-\frac1{n+1}} u, \\[2ex]
\tilde f=\dfrac{\delta_1{\delta_3}^n}{\Delta^2}|\delta_4 x+\delta_5|^{\frac{3n+4}{n+1}} f,
\quad \tilde
h=\dfrac{{\delta_3}^{n-m+1}}{\Delta^2}|\delta_4 x+\delta_5|^{\frac{m+3n+3}{n+1}} h, \quad
\tilde n=n, \quad \tilde m=m,
\end{array}
\]
where $\delta_j,$ $j=1,\dots,7$, are arbitrary constants such that
$\delta_1\delta_3\not=0$, $\Delta=\delta_5\delta_6-\delta_4\delta_7\ne0$ and
the tuple $(\delta_4,\delta_5,\delta_6,\delta_7)$ is defined up to a nonzero multiplier, 
e.g., we can set $\Delta=\pm1$.
The arbitrary element~$n$ is assumed to be unequal to~$-1$.
For $n=-1$ transformations from the group~$G^{\sim}_{g=1}$ take the form
\[
\begin{array}{l}
\tilde t=\delta_1 t+\delta_2,\quad \tilde x=\delta_4 x+\delta_5,
\quad
\tilde u=\delta_3e^{\delta_6 x} u, \\[1ex]
\tilde f=\dfrac{\delta_1}{{\delta_4}^2\delta_3e^{\delta_6 x}}  f,
\quad \tilde h=\dfrac{1}{{\delta_4}^2{\delta_3}^m e^{m\delta_6 x}}
h, \quad \tilde n=n, \quad \tilde m=m,
\end{array}
\]
where  $\delta_j,$ $j=1,\dots,6$, are arbitrary constants, $\delta_1\delta_3\delta_4\not=0$.
\end{theorem}

\begin{theorem}\label{CorollaryOnEquivGroupMN1}
The conditional equivalence group~$G^{\sim}_{g=1,\,m=n+1}$ of class~\eqref{eqRDfghTwoPower}
associated with the constraints $g=1$ and $m=n+1$ is formed by the transformations
\begin{gather*}
\tilde t=\delta_1 t+\delta_2,\quad \tilde x=\varphi(x),\quad\tilde u=\psi(x)u,  \\
\tilde f=\frac{\delta_0^2\delta_1}{|\psi|^{3n+4}}f,\quad \tilde
h=\delta_0^2\frac{h-|\psi|^{n+1}[|\psi|^{-(n+2)}\psi_x]_x}{|\psi|^{4n+4}},\quad
\tilde n=n.
\end{gather*}
Here~$\varphi$ and~$\psi$ are arbitrary smooth functions of~$x$
and $\delta_j$, $j=0,1,2$, are arbitrary constants
satisfying the conditions $\delta_0\varphi_x=\psi^{2n+2}$ and $\delta_0\delta_1\psi\ne0$.
\end{theorem}

In the course of group classification in the general case $m\ne1,n+1$ we derive that 
the maximal Lie invariance algebra of an equation from class~\eqref{eqRDfghTwoPower} with $g=1$ 
is a proper extension of the kernel algebra $\langle\p_t\rangle$ only if 
the corresponding value of the arbitrary elements~$f$ and~$h$ satisfy a system of the form 
\begin{gather}\label{eqRDfghTwoPowerABCDP}
\begin{split}
&((n+1)ax^2+bx+c)\dfrac{f_x}f=-(3n+4)ax+d,\\
&((n+1)ax^2+bx+c)\dfrac{h_x}h=-(3(n+1)+m)ax+(1+n-m)p-2b,
\end{split}
\end{gather}
where $a$, $b$, $c$, $d$ and $p$ are constants which are not simultaneously equal to zero.  

\begin{lemma}\label{LemmaOntransOfCoeffsOfClassifyingSystem2}
If $n\ne-1$, up to $G^{\sim}_{g=1}$-equivalence
the parameter tuple~$(a,b,c,d,p)$ can be assumed to belong to the set
\[
\{(0,1,0,\bar d,\bar p),\ (0,0,1,1,\check p),\ (0,0,1,0,1),\ (0,0,1,0,0),\ (1/(n+1),0,1,\hat d,\hat p)\},
\]
where $\hat d\geqslant0$ and, if $\hat d=0$, $\hat p\geqslant0$; $\check p$ is an arbitrary constant; 
\begin{equation}\label{inequalities_dp}
\bar d\geqslant-\frac{3n+4}{2(n+1)}\quad\mbox{and,\quad if\quad}
\bar d=-\frac{3n+4}{2(n+1)},\quad \bar p\geqslant\frac{1}{2(n+1)}.
\end{equation}
If $n=-1$, $G^{\sim}_{g=1}$-inequivalent values of the parameter tuple~$(a,b,c,d,p)$ 
are exhausted by elements of the set
\[
\{(0,1,0,d',p'),\ (0,0,1,0,1),\ (0,0,1,0,0),\ (\varepsilon'',0,1,0,p'')\},
\]
where
$d'$ and $p'$ are arbitrary constants, $\varepsilon''=\pm1$ and $p''\geqslant0$.
\end{lemma}

\begin{proof}
Combined with the multiplication by a nonzero constant,
each transformation from the equivalence group~$G^{\sim}_{g=1}$ is extended to the coefficient tuple
of the system~\eqref{eqRDfghTwoPowerABCDP}. 
The extension takes the form 
\begin{gather*}
\begin{array}{l}
(n+1)\tilde a=\nu(\delta_5^2(n+1)a-\delta_4\delta_5b+\delta_4^2c),\\[1ex]
\tilde b=\nu(-2\delta_5\delta_7(n+1)a+(\delta_4\delta_7+\delta_5\delta_6)b-2\delta_4\delta_6c),
\\[1ex]
\tilde c=\nu(\delta_7^2(n+1)a-\delta_6\delta_7b+\delta_6^2c),
\quad\tilde d=\nu\Delta d+\dfrac{3n+4}{n+1}\nu(\delta_5\delta_7(n+1)a-\delta_4\delta_7b+\delta_4\delta_6c),\\[1ex]
\tilde p=\nu\Delta p-\dfrac1{n+1}\nu(\delta_5\delta_7(n+1)a-\delta_4\delta_7b+\delta_4\delta_6c)
\end{array}
\end{gather*}
if $n\neq-1$ and 
\begin{gather}\label{eq_tr_abcdp_pow_n=-1}
\begin{array}{l}
\tilde a=\dfrac {\nu}{\delta_4}{(a+\delta_6b)},\quad \tilde b
=\nu b,\quad \tilde c=\nu({\delta_4c-\delta_
5b}),\\[1ex]  \tilde d=\nu d+\dfrac{\nu}{\delta_4}(\delta_5a+\delta_5\delta_6b-\delta_6\delta_4c),
\quad
\tilde p=\nu p-\dfrac{\nu}{\delta_4}(\delta_5a+\delta_5\delta_6b-\delta_6\delta_4c)
\end{array}
\end{gather}
if $n=-1$.
Here $\Delta=\delta_5\delta_6-\delta_4\delta_7\ne0$ and $\nu$ is an arbitrary nonzero constant.

If $n\neq-1$, the proof is similar to the proof of Lemma~\ref{LemmaOntransOfCoeffsOfClassifyingSystem}. 
In this case there are only three $\hat G^\sim_1$-inequivalent values of the triple $(a,b,c)$
depending on the sign of $D=b^2-4nac$,
\begin{gather*}
(0,1,0)\quad\mbox{if}\quad D>0, \quad
(0,0,1)\quad\mbox{if}\quad  D=0,\quad
(1/(n+1),0,1)\quad\mbox{if}\quad  D<0.
\end{gather*}
Indeed, if $D>0$ then there exist two linearly independent pairs
$(\delta_4,\delta_5)$ and $(\delta_6,\delta_7)$
such that $\delta_5^2(n+1)a-\delta_4\delta_5b+\delta_4^2c=0$ and $\delta_7^2(n+1)a-\delta_6\delta_7b+\delta_6^2c=0$.
For these values of $\delta$'s we have $\tilde a=\tilde c=0$.
In the case $D=0$ we choose values of $\delta_4$, $\delta_5$, $\delta_6$ and $\delta_7$
for which $\delta_5^2(n+1)a-\delta_4\delta_5b+\delta_4^2c=0$ and
the pair $(\delta_6,\delta_7)$ is not proportional to the pair $(\delta_4,\delta_5)$.
Then we obtain that $\tilde a=0$ and
$\tilde b=\nu\delta_7(\delta_4b-2\delta_5(n+1)a)+\nu\delta_6(\delta_5b-2\delta_4c)=0$.
The residual coefficient ($\tilde b$ if $D>0$ and $\tilde c$ if $D=0$) is necessarily nonzero
and hence it can be scaled to~$1$ using the multiplication by the appropriate value of~$\nu$.
If $D<0$, we have $ac\ne0$ and can set $a>0$.
As the matrix \[\left(\begin{array}{cc}(n+1)a&-b/2\\-b/2&c\end{array}\right)\]
is symmetric and positive, 
the corresponding bilinear form is a well-defined scalar product.
Choosing $\nu=1$ and pairs $(\delta_4,\delta_5)$ and $(\delta_6,\delta_7)$
which are orthonormal with respect to this product, 
we obtain $(n+1)\tilde a=\tilde c=1$ and $\tilde b=0$.

Now for each of the above inequivalent form
for both the tuples $(a,b,c)$ and $(\tilde a,\tilde b,\tilde c)$
we look for possible gauges of the coefficients~$d$ and~$p$.

Thus, from $(a,b,c)=(\tilde a,\tilde b,\tilde c)=(0,1,0)$ we derive
$\delta_4\delta_5=\delta_6\delta_7=0$ and $\delta_4\delta_7+\delta_5\delta_6=\nu^{-1}$.
This system in~$\delta$'s has two solutions, 
$\delta_4=\delta_7=0$ with $\Delta=\delta_5\delta_6=\nu^{-1}$ and 
$\delta_5=\delta_6=0$ with $\Delta=-\delta_4\delta_7=-\nu^{-1}$.
The first solution leads to the identical transformation of the coefficients~$d$ and~$p$.
For the second solution the transformation takes the form $\tilde d=-d-(3n+4)/(n+1)$, 
$\tilde p=-p+1/(n+1)$.
This is why up to $G^{\sim}_{g=1}$-equivalence we can assume
that $\tilde d$ and  $\tilde p$ satisfy~\eqref{inequalities_dp}.

Setting $(a,b,c)=(\tilde a,\tilde b,\tilde c)=(0,0,1)$ results in $\delta_4=0$, $\nu\delta_6^2=1$
and $\Delta=\delta_5\delta_6\ne0$.
The transformation of $d$ and $p$ is reduced to simultaneous scaling with the same multipliers,
$\tilde d=\nu\Delta d$ and $\tilde p=\nu\Delta p$. 
This allows us either to set $\tilde d=1$ if $d\ne0$ or 
to scale~$\tilde p$ if $d=0$ and hence $\tilde d=0$. 
As a result, we obtain the tuples $(0,0,1,1,\check p)$ if $d\neq0$,  
$(0,0,1,0,1)$ if $d=0$ and $p\neq0$, and $(0,0,1,0,0)$ in the case $d=p=0.$ 
The last tuple corresponds to equations of the form~\eqref{eqRDfghTwoPower} with constant coefficients 
(Case 3 of Table 2).

The equality $((n+1)a,b,c)=((n+1)\tilde a,\tilde b,\tilde c)=(1,0,1)$ implies 
$\delta_4^2+\delta_5^2=\delta_6^2+\delta_7^2=\nu^{-1}$ and $\delta_4\delta_6+\delta_5\delta_7=0$.
Hence $\delta_6=\tilde\varepsilon\delta_5$ and $\delta_7=-\tilde\varepsilon\delta_4$, 
where $\tilde\varepsilon=\pm 1$.
The transformation of the coefficients~$d$ and~$p$ is reduced to the multiplication by~$\tilde\varepsilon$,
$\tilde d=\tilde\varepsilon d$ and $\tilde p=\tilde\varepsilon p$.
This is why we can only set
$\hat d\geqslant0$ and, if $\hat d=0$, $\hat p\geqslant0$.

Consider the case $n=-1$. 
As the tuple $(a,c,b,d,p)$ is nonzero, system~\eqref{eq_tr_abcdp_pow_n=-1} implies that $(b,c)\ne(0,0)$. 
If the tuple $(a,b,c)$ satisfies the condition $b\ne0$ (resp. $b=a=0$, resp. $b=0$ and $a\ne0$) 
then it is $G^{\sim}_{g=1}$-equivalent to the tuple $(0,1,0)$ (resp. $(0,0,1)$, resp. $(\varepsilon'',0,1)$). 
Similarly to the case $n\neq-1$, now we look for possibilities of gauging the parameters~$d$ and~$p$
after fixing one of the above inequivalent forms for both the tuples $(a,b,c)$ and $(\tilde a,\tilde b,\tilde c)$. 
 
The equality $(a,b,c)=(\tilde a,\tilde b,\tilde c)=(0,1,0)$ implies $\nu=1$ and $\delta_5=\delta_6=0.$
Therefore, the parameters $d$ and $p$ are identically transformed.

Setting $(a,b,c)=(\tilde a,\tilde b,\tilde c)=(0,0,1)$ in~\eqref{eq_tr_abcdp_pow_n=-1} results in
$\delta_4=\nu^{-1}$, $\tilde d=\nu(d-\delta_6)$ and $\tilde p=\nu(p+\delta_6)$. 
Therefore, we can set $(\tilde d,\tilde p)$ to be equal $(0,0)$ or $(0,1)$. 

It follows from~\eqref{eq_tr_abcdp_pow_n=-1} with $(a,b,c)=(\tilde a,\tilde b,\tilde c)=(\varepsilon'',0,1)$, 
where $\varepsilon''=\pm1$,  
that $\nu=\delta_4=\pm1$, $\tilde d=\nu d+\varepsilon''\delta_5-\nu\delta_6$ 
and $\tilde p=\nu p-\varepsilon''\delta_5+\nu\delta_6$.
This allows for setting one of the parameters $d$ and $p$ to zero. 
After fixing the zero value of~$d$, we can additionally alternate the sign of $\tilde p$.
Hence we assume $(\tilde d,\tilde p)=(0,p'')$, where $p''\geqslant0$.
\end{proof}

Lemma~\ref{LemmaOntransOfCoeffsOfClassifyingSystem2} implies that
up to $G^\sim_{g=1}$-equivalence Case 2 of Table~2 is partitioned into three inequivalent subcases,
\begin{enumerate}\itemsep=0ex
\item 
$(f,h)=(|x|^d,\varepsilon|x|^q), \ q=(1+n-m)p-2\colon\quad$ $\langle\partial_t,\,(d+2-pn)t\partial_t+
x\partial_x+pu\partial_u\rangle$;
\item 
$(f,h)=(e^{dx}, \varepsilon e^{qx}), \ q=(1+n-m)p\colon\quad$
$\left\langle\partial_t,\,\left(d-pn\right)t\partial_t+\partial_x+pu\partial_u\right\rangle$;
\item 
$n\neq -1$, $(f,h)=\left((x^2+1)^{-\frac{3n+4}{2(n+1)}}e^{d\arctan x},\varepsilon (x^2+1)^{-\frac{3n+3+m}{2(n+1)}}e^{q\arctan x}\right),
 \ q=(1+n-m)p\colon$\\
$\langle\partial_t,\,(d-pn)t\partial_t+(x^2+1)\partial_x+(x/(n+1)+p)u\partial_u\rangle$;\\[1ex]
$n=-1$, $(f,h)=\left(e^{\frac12\varepsilon''x^2},\varepsilon e^{\frac{m}2\varepsilon''x^2+qx}\right), \ q=-mp\colon$
$\langle\partial_t,\,qt\partial_t-m\partial_x+(q+m\varepsilon''x)u\partial_u\rangle$. 
\end{enumerate}
Here $\varepsilon,\varepsilon''=\pm1$, $(d,q)\neq(0,0)$ for the first and second subcases 
and additionally $(d,q)\neq\left(-(3n+4)/(n+1),-3-m/(n+1)\right)$ in the first subcase with $n\ne-1$ 
and $(d,q)\neq(0,0)$ in the third subcase with $n=-4/3$; 
otherwise we have cases of Lie symmetry extensions of greater dimensions. 
It follows from Lemma~\ref{LemmaOntransOfCoeffsOfClassifyingSystem2} that
up to $G^\sim_{g=1}$-equivalence we can set certain constraints for the parameters~$d$ and~$q$.
(It is convenient to use~$q$ instead of~$p$ as a parameter.) 
These constraints are different for the cases $n\ne-1$ and $n=-1$.

If $n\neq-1$, for the first subcase an exhaustive gauge implied by $G^\sim_{g=1}$-equivalence
consists of the inequalities 
\[
d\geqslant-\frac{3n+4}{2(n+1)}\quad\mbox{and,\quad if}\quad d=-\frac{3n+4}{2(n+1)},\quad 
q\geqslant-\frac{3n+3+m}{2(n+1)}.
\]
They can be set using the equivalence transformation
\[
\tilde t=t,\quad\tilde x=\frac1x,\quad \tilde u=|x|^{-\frac1{n+1}}u,
\]
whose extension to the parameters~$d$ and $q$ is given by 
$
\tilde d=-d-(3n+4)/(n+1)
$ 
and
$
\tilde q=-q-3-m/(n+1).
$
In the second subcase the parameters~$d$ and $q$ can be gauged using a scaling of $x$.
More precisely, $d=1\bmod G^\sim_{g=1}$ if $d\neq0$ and $q=1\bmod G^\sim_{g=1}$ if $d=0$.
In the last subcase we can just simultaneously alternate the signs of~$d$ and~$q$.
Hence the exhaustive gauge is presented by $d\geqslant0$ and, if $d=0$, $q\geqslant0$.

Consider $n=-1$. Then in the first subcase 
the parameters $d$ and $q$ are not changed by transformations from $G^\sim_{g=1}$. 
In the second subcase the parameters~$d$ and $q$ can be gauged 0 and~1, respectively. 
In the last subcase we can alternate the sign of~$q$ and assume $q\geqslant0\bmod G^\sim_{g=1}$.

A similar partition can be also carried out for Case~5 of Table~2. 
Namely, we have the following inequivalent subcases:

\begin{enumerate}\itemsep=0ex
\item 
$f=|x|^d\colon\quad$ 
$\langle\partial_t,\,e^{-\varepsilon nt}(\partial_t+\varepsilon u\partial_u),\,
nx\partial_x+(d+2)u\partial_u\rangle$;
\item 
$f=e^x\colon\quad$
$\langle\partial_t,\,e^{-\varepsilon nt}(\partial_t+\varepsilon u\partial_u),\,
n\partial_x+u\partial_u\rangle$;
\item 
$n\neq -1$, 
$f=(x^2+1)^{-\frac{3n+4}{2(n+1)}}e^{d\arctan x}\colon$\\
$\langle\partial_t,\,e^{-\varepsilon nt}(\partial_t+\varepsilon u\partial_u),\,
n(x^2+1)\partial_x+(nx/(n+1)+d)u\partial_u\rangle$;

$n=-1$, $f=e^{\frac12\varepsilon''x^2}\colon\quad$
$\langle\partial_t,\,e^{\varepsilon'' t}(\partial_t+\varepsilon'' u\partial_u),\,
\partial_x+(\varepsilon'' x-d)u\partial_u\rangle$.
\end{enumerate}
In the first subcase $d\ne0$ and, if $n\ne-1$, $d\ne-(3n+4)/(n+1)$.
In the third subcase $d\ne0$ if $n=-4/3$ and $\varepsilon''=\pm1$ if $n=-1$. 
The gauges for the parameter~$d$ modulo $G^\sim_{g=1}$-equivalence  
coincide with the gauges of~$d$ in the respective subcases of Case~2 from~Table~2.

\subsection*{Acknowledgements}

OOV and ROP express their gratitude to the hospitality
shown by University of Cyprus during their visits to the University.
The research of ROP was supported by the Austrian Science Fund (FWF), project P20632.

\end{document}